\documentclass{article}
\usepackage{graphicx} % Required for inserting images
\usepackage{xcolor}
\usepackage{amsmath}
\usepackage{amssymb}
\usepackage{amsthm}
\usepackage{cleveref}
\usepackage{a4wide}

\newtheorem{theorem}[]{Theorem}
\newtheorem{lemma}[theorem]{Lemma}

\newcommand{\R}{\mathbb{R}}

\newcommand{\C}{\mathcal{C}}

\newcommand{\OPT}{\mathrm{OPT}}

\newcommand{\CLP}{\mathrm{CLP}}
\newcommand{\SEP}{\mathrm{SEP}}
\newcommand{\LINSEP}{\mathrm{LINSEP}}
\newcommand{\EX}{\mathrm{X}}
\newcommand{\DW}{\mathrm{CLP}^{\mathrm{DW}}}
\newcommand{\BLP}{\mathrm{B}}

\usepackage{tcolorbox}
\newtcolorbox{LPBox}[1][]{
colframe=black!15,
colback=white,
coltitle=black,
before upper=\setlength{\parindent}{0em}\everypar{{\setbox0\lastbox}\everypar{}},
title    = {#1}
}

\usepackage[colorinlistoftodos]{todonotes}
\newcommand{\noteSM}[1]{\todo[color = orange, inline]{\footnotesize SM: #1}{}}
\renewcommand{\epsilon}{\varepsilon}

\title{The Submodular Santa Claus Problem}
\author{\'Etienne Bamas\footnote{Post-Doctoral Fellow, ETH AI Center, Switzerland. etienne.bamas@inf.ethz.ch} \and Sarah Morell\footnote{TU Berlin, Germany. morell@math.tu-berlin.de. Funded by the Deutsche Forschungsgemeinschaft (DFG, German Research Foundation) under Germany´s Excellence Strategy – The Berlin Mathematics Research Center MATH+ (EXC-2046/1, project ID: 390685689).} \and Lars Rohwedder\footnote{Maastricht University, Netherlands. l.rohwedder@maastrichtuniversity.nl. Supported by Dutch Research Council
(NWO) project “The Twilight Zone of Efficiency: Optimality of Quasi-Polynomial Time Algorithms” [grant number
OCEN.W.21.268].}}
\date{}

\begin{document}

\maketitle

\begin{abstract}
    We consider the problem of allocating indivisible resources to players
    so as to maximize the minimum total value any player receives.
    This problem is sometimes dubbed the Santa Claus problem
    and its different variants have been subject to extensive research
    towards approximation algorithms over the past two decades.

    In the case where each player
    has a potentially different additive valuation function,
    Chakrabarty, Chuzhoy, and Khanna [FOCS'09] gave an $O(n^{\epsilon})$-approximation algorithm
    with polynomial running time for any constant $\epsilon > 0$ and a polylogarithmic approximation algorithm in quasi-polynomial time.
    We show that the same can be achieved for monotone submodular valuation
    functions, improving over the previously best algorithm due
    to Goemans, Harvey, Iwata, and Mirrokni [SODA'09], which has
    an approximation ratio of more than $\sqrt{n}$.

    Our result builds up on a sophisticated LP relaxation,
    which has a recursive block structure that allows us to solve it
    despite having exponentially many variables and constraints.
\end{abstract}

\section{Introduction}
The egalitarian welfare is
the value of the least happy player. Other
natural welfare functions include utilitarian welfare, the sum of values, and Nash social welfare, the product of values.
Egalitarian welfare can
be seen as the trade-off that emphasizes most extremely on fairness.
In this paper we study the problem of allocating indivisible resources
with the goal of maximizing egalitarian welfare, which
is sometimes called the Santa Claus problem or max-min fair allocation.

\paragraph{Problem setting.}
Given resources $R$ and players $P$, we
wish to find an allocation $\sigma: R\rightarrow P$ such
that
\begin{equation*}
    \min_{p\in P} \ f_p(\{r \in R : \sigma(r) = p\})
\end{equation*}
is maximized. Here $f_p: 2^R \rightarrow \R_{\geq 0}$ with $p \in P$ is the function that specifies the
value that player $p$ receives from a set of resources.
When considering polynomial-time approximation algorithms,
one typically assumes the function can be accessed by value queries,
i.e., an oracle that returns $f_p(A)$ for some given $p$ and $A$ in polynomial time.
Strong assumptions on the functions are necessary in order to
hope for any meaningful algorithmic guarantees.
At the same time, the functions
should still remain expressive enough to capture realistic scenarios.

A natural assumption is that each $f_p$ is non-negative and additive (a linear function), which means that there are values $v_{p,r} \in\mathbb R_{\ge 0}$
for each $p\in P, r\in R$ and $f_p(A) = \sum_{r\in A} v_{p, r}$ for
each $A\subseteq R$.
Already for this class of functions, the study of approximation
algorithms for the Santa Claus problem has proven to be extremely
challenging. 
The best algorithm due to Chakrabarty, Chuzhoy, and Khanna~\cite{DBLP:conf/focs/ChakrabartyCK09}
achieves an $n^{\epsilon}$-approximation
in polynomial time, for every fixed constant $\epsilon > 0$,
or a $(\log^{O(1)}n)$-approximation in quasi-polynomial time,
more precisely in time $n^{O(\log n/\log \log n)}$.
The best lower bound on the approximation ratio
of a polynomial-time algorithm is $2$ (assuming \texttt{P}$\neq$\texttt{NP}).
Closing the gap remains a big open question in approximation algorithms
and is connected to the similar task of minimizing
makespan on unrelated parallel machines, see~\cite{bamas2024santa}.

An important generalization of additive functions is the class of monotone submodular functions. A function $f$ is submodular if it satisfies
the \emph{diminishing marginal returns property}, which means that
\[f(A \cup \{r\}) - f(A) \ge f(B \cup \{r\}) - f(B) \text{ for all } A\subseteq B \text{ and } r\notin B.\] 
Monotonicity states that $f(A) \le f(B)$ for $A\subseteq B$.
As is standard in approximation algorithms, we also assume that
any monotone submodular function is normalized such that
$f(\emptyset) = 0$.
Staying the metaphor of Santa Claus, an example of the diminishing marginal returns
property is that the value of an apple is higher for a child, when the child does not have a donut than when it does.

For utilitarian welfare or Nash social welfare, the class of monotone submodular functions still admits good algorithms~\cite{li2022constant, vondrak2008optimal}, which raises the
question whether the restriction to additive functions is necessary
in egalitarian welfare.

Even more general than monotone submodular functions are
subadditive functions, which only need to
satisfy $f(A \cup B) \le f(A) + f(B)$
for all $A, B$. 
Next to additive functions, submodular and subadditive functions
are arguably the most fundamental classes of valuation functions.
Both submodular and subadditive functions
are also briefly mentioned by Chakrabarty, Chuzhoy, and Khanna~\cite{DBLP:conf/focs/ChakrabartyCK09}
who emphasize that at the time the best lower bound for both in
the Santa Claus problem was also only $2$. Since then however, 
Barman, Bhaskar, Krishna, and Sundaram~\cite{barman_et_al:LIPIcs.ESA.2020.11}
have proven that for
XOS functions, a class of functions that lies between monotone submodular and subadditive, any $O(n^{1-\epsilon})$-approximation algorithm
requires exponentially many value queries, which therefore also forms a lower bound on subadditive functions.
This lower bound is in the value query model. In literature there are
more powerful query models, for example demand queries, see e.g.~\cite{barman_et_al:LIPIcs.ESA.2020.11, feige2009maximizing},
which we do not detail here.

Monotone submodular functions are not subject to the mentioned
lower bound and a sublinear approximation rate is
indeed possible: Goemans, Harvey, Iwata, and Mirrokni~\cite{goemans2009approximating} gave a reduction to the additive case,
which loses a factor of $O(\sqrt{n}\log n)$ and thus leads
to an $O(n^{1/2 + \epsilon})$-approximation algorithm when combined
with~\cite{DBLP:conf/focs/ChakrabartyCK09}.

\paragraph{Contribution and outline.} 
In this paper we achieve a direct generalization from the
additive to the monotone submodular case, 
without the reduction of Goemans, Harvey, Iwata, and Mirrokni~\cite{goemans2009approximating}, and match the state-of-the-art from the additive case.

\begin{theorem}
    \label{thm:main}
    For the Submodular Santa Claus problem there is a polylogarithmic approximation algorithm with running time $n^{O(\log n / \log\log n)}$ and
    an $n^{\epsilon}$-approximation algorithm with running time $n^{O(1/\epsilon)}$ for any constant $\epsilon > 0$.
\end{theorem}
Similar to Chakrabarty, Chuzhoy, and Khanna~\cite{DBLP:conf/focs/ChakrabartyCK09}, our techniques can be divided into three steps:
\begin{enumerate}
    \item Reducing to a carefully designed layered flow problem, the \emph{augmentation problem}.
    \item Formulating and solving a strong linear programming relaxation of the augmentation problem.
    \item Obtaining integral solution of the augmentation problem via randomized rounding.
\end{enumerate}
The most challenging part to generalize is the second step. In order
to include submodular valuation functions in the linear programming
relaxation, we use the standard concept of configuration variables,
i.e., a variable for each set of resources that has a sufficiently large value.
The natural way of writing a configuration LP, however, is not 
sufficient even for the additive case, see e.g.~\cite{bansal2006santa}.
The formulation used by Chakrabarty, Chuzhoy, and Khanna for the additive case is highly non-trivial. It is closely related to using a constant or logarithmic number of rounds of the Sherali-Adams hierarchy on a naive formulation, see~\cite{bateni2009maxmin}. Notably, their linear program is strong
for the augmentation problem,
to which they give a non-trivial reduction (step~1),
but not for the original problem.

Combining their approach with configuration variables as above
leads to a linear program that has both an exponential number of variables
and constraints, an issue that does not occur in the additive case. Typically, one needs to have either a polynomial number
of variables or constraints in order to even 
hope to solve a linear program efficiently. Otherwise,
already the encoding of a solution might require exponential space.
Exceptions are very rare, see e.g.~\cite{DBLP:journals/corr/abs-2211-08381}.
The distinct feature of our linear programming relaxation is
that it has a recursive block structure. Specifically,
the matrix consists of an exponentially large number of blocks along the diagonal which are linked by a polynomial number of constraints.
The blocks on the diagonal exhibit the same structure recursively up
to a recursion depth of $h$,
see Figure~\ref{fig:block}.
It can be shown that if in addition 
the feasible region of each block (ignoring the linking constraints)
forms a polyhedral cone, then
indeed such matrices always have solutions
of support $n^{O(h)}$ (if any exist).
This is then polynomial for constant recursion depth $h$.
We solve our specific formulation using ideas from the
Dantzig-Wolfe decomposition~\cite{dantzig1960decomposition},
where the pricing problem requires
a combination of recursively solving a linear program with
lower recursion depth and
the search for a configuration of high submodular function value.
For the latter we use the multilinear extension and continuous Greedy
in a non-standard variant.
This way we arrive at a sufficiently good approximation of the
linear program.

\begin{figure}
    \centering
    \begin{tikzpicture}[scale=0.8]
        \draw[fill=lightgray] (0, 10.5) rectangle (4.5, 10);
        \draw[fill=lightgray] (0, 10) rectangle (1.5, 9.5);
        \draw[fill=lightgray] (0, 9.5) rectangle (0.5, 9);
        \draw[fill=lightgray] (0.5, 9) rectangle (1, 8.5);
        \node at (1.25, 8.35) {$\ddots$};
        \draw[fill=lightgray] (1.5, 8) rectangle (3, 7.5);
        \draw[fill=lightgray] (1.5, 7.5) rectangle (2, 7);
        \draw[fill=lightgray] (2, 7) rectangle (2.5, 6.5);
        \node at (2.75, 6.35) {$\ddots$};
        \node at (3.75, 5.85) {$\ddots$};
        \draw[thick] (-0.2, 10.5) to[bend right=10] (-0.2, 5.5);
        \draw[thick] (4.7, 10.5) to[bend left=10] (4.7, 5.5);

        \draw[|-|] (-1, 10.5) to node[pos=0.5, left] {poly} (-1, 10);
        \draw[|-|] (5.5, 10.5) to node[pos=0.5, right] {exp} (5.5, 5.5);
        \draw[|-|] (0, 5) to node[pos=0.5, below] {exp} (4.5, 5);
    \end{tikzpicture}
    \caption{Block structure of non-zero entries in constraint matrix of linear programming relaxation}
    \label{fig:block}
\end{figure}
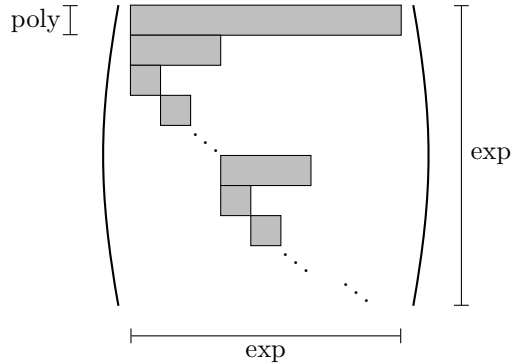

As an additional contribution, we significantly simplify the reduction step of Chakrabarty, Chuzhoy, and Khanna which used very intricate techniques and non-trivial graph theory results. In contrast, our reduction only uses standard flow arguments.
\paragraph{Further related work.} If all functions are identical and additive, then
there exists a PTAS for the Santa Claus problem~\cite{woeginger1997polynomial, epstein1999approximation}.
If they are identical and monotone submodular,
then a Greedy type of algorithm still achieves a constant approximation~\cite{krause2009simultaneous}.

A substantially harder variant, which has received a lot of attention
is the so-called restricted assignment case.
Here, the valuation functions are identical ($f_p = f_{p'}$
for $p, p'\in P$), but each player $p$ can only receive
resources from a specific set $R(p) \subseteq R$.
This can equivalently be phrased as $f_p(A) = f(A \cap R(p))$
for some uniform function $f$, or in the additive case, that
$v_{p, r}\in \{0,v_r\}$ for some uniform value $v_r$ for each resource $r\in R$.
Most of this work focuses on the mentioned additive case,
leading to a constant factor approximation
for this case~\cite{bansal2006santa,feige2008allocations,annamalai2017combinatorial,davies2020tale,PolacekS2015,DBLP:journals/talg/AsadpourFS12}. For the submodular case, an $O(\log\log n)$-approximation algorithm is known~\cite{bamas2021submodular}.
These works heavily rely on the configuration LP, a linear
programming relaxation, which is known to have a high integrality
gap outside of the restricted assignment variant~\cite{bansal2006santa}.
Therefore these techniques have only limited impact towards the
goals of this paper.

Outside the restricted assignment problem, some progress
towards a constant approximation is due to Bamas and Rohwedder~\cite{bamas2023better}
who gave a $(\log^{O(1)}\log n)$-approximation algorithm in quasi-polynomial time
for the so-called max-min degree arborescence problem,
a special case of the additive variant, where the configuration
LP already has a high integrality gap.

The dual problem of minimizing the maximum instead of maximizing the minimum
function value is
usually motivated by machine scheduling, specifically makespan minimization.
Here, the additive case is well known to admit a constant approximation~\cite{DBLP:journals/mp/LenstraST90} and
the reduction by Goemans, Harvey, Iwata, and Mirrokni~\cite{goemans2009approximating} works in the same way, yielding a polynomial-time $O(\sqrt{n}\log n)$-approximation algorithm
for makespan minimization with monotone submodular load functions.
Interestingly,
this is known to be the best possible up to logarithmic factors
in the value query model, see~\cite{svitkina2011submodular},
and therefore behaves differently to the problem studied
in this paper.

\section{Algorithmic framework}
In this chapter, we introduce the augmentation problem as well as the linear programming relaxation for it.
Those are the pillars of our algorithm that connect the three steps outlined
in the introduction.
In \Cref{sec:Reduction} we then show how to reduce to the augmentation problem, in \Cref{sec:linearprogram} we explain how to solve its linear programming relaxation, and in \Cref{sec:rounding} we present the rounding algorithm for the
relaxation.

\subsection{The augmentation problem}
\label{subsec:intro-reduction}
As the name suggests, the augmentation problem is related to 
augmenting some partial solution of the Submodular Santa Claus problem to
a better solution. This can be seen as a much more involved variant of
finding augmenting paths to solve bipartite matching.
Similar to there, we will later invoke it several times in order to arrive
at the final solution for the Santa Claus problem.
We formulate the augmentation problem in purely graph theoretical terms.

An instance of the augmentation problem contains several levels.
We will start by introducing the structure within one level.

\paragraph*{One level of the augmentation problem.}
Let $G = (V, E)$ be a directed graph and let $S, T\subseteq V$ denote disjoint sets of sources and
sinks. Each source in $S$ has exactly one outgoing edge and no incoming edges.
Each sink in $T$ has only incoming edges.
Furthermore, for all $v\in T$ let $f_v : 2^{\delta(v)} \rightarrow \mathbb R_{\ge 0}$ be monotone, submodular functions. Here, $\delta(v)$ are the edges incident to $v$.

The solution for this level is a binary flow $g : E \rightarrow \{0, 1\}$ from
$S$ to $T$, i.e.,
flow conservation is satisfied on $V\setminus (S\cup T)$.
We write $E(g) = \{e\in E : g(e) > 0\}$ and $V(g) = \bigcup_{(u, v)\in E(g)} \{u, v\}$. Furthermore, in slight abuse of notation we write $V' \cap g = V' \cap V(g)$ for some $V'\subseteq V$ and $E' \cap g = E'\cap E(g)$ for some $E'\subseteq E$.
We say that sink $v\in T$ is $\alpha$-covered
by $g$ if $f_v(g \cap \delta(v)) \ge 1/\alpha$. We give an example in Figure \ref{fig:augmentation_one_level}.

\begin{figure}[ht]
    \centering
    \includegraphics{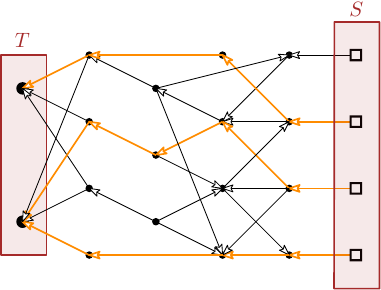}
    \caption{An instance of the one-level augmentation problem. The top sink on the left has valuation function equal to the total flow received, and the sink at the bottom has valuation function equal to the total flow received divided by $2$. The set of bold orange edges forms a feasible solution which covers all the sinks in $T$.}
    \label{fig:augmentation_one_level}
\end{figure}

\paragraph*{Augmentation problem.}
For $h \in \mathbb N$, an $h$-level instance of the augmentation problem consists
of levels $(G_i, S_i, T_i)$ with $G_i = (V_i, E_i)$ for $i=1,2,\dotsc,h$
with the structure as above and monotone submodular
functions $f_v : 2^{\delta(v)} \rightarrow \mathbb R_{\ge 0}$ for
$v\in T_1\cup T_2\cup \dotsc\cup T_h$. In addition, there are \emph{linking edges} $L_i \subseteq  T_{i+1}\times S_i$ for $i = 1,2,\dotsc, h-1$.
Each set $L_i$ forms a matching, i.e., the edges are disjoint.
For $U\subseteq S_i$ we write $L_i(U) = \{v\in T_{i+1} : (v, u)\in L_i \text{ for some } u\in U\}$.

A solution consists of a flow $g_i$ for each level $i$, as
described above.
The $h$ levels depend on each other in
that for the source $g \cap S_i$ is used
by the flow in level $i$, $g_{i+1}$  must cover $L_i(g \cap S_i)$.
In other words, if $s\in g\cap S_i$ and there is no edge $(u, s)\in L_i$ for any $u\in T_{i+1}$, i.e., $L_i(\{s\}) = \emptyset$, then there is no further requirement and
if indeed there exists such an edge then
we may informally think of the flow leaving source $s$ to arrive through the linking edge and indirectly from
$u$. Note, however, that $u$ may require an incoming
flow higher than the amount of flow leaving $s$.
The sinks of the first level $T_1$ and the sources of the last level $S_h$ have no dependencies with other levels.

To conclude the description of the problem, a solution
is feasible for some $T^*\subseteq T_1$ if
\begin{itemize}
    \item $g_1$ $\alpha$-covers each $v\in T^*$ and 
    \item for $i=1,2,\dotsc,h-1$, solution $g_{i+1}$ $\alpha$-covers each $v \in L_i (g_i\cap S_i)$.
\end{itemize}
We refer the reader to Figure \ref{fig:augmentation_multi_level} for an example.
In the remainder we denote by $n$ the total number of vertices in all $h$ graphs. This will be polynomial in the size
of the original instance. 

\begin{figure}[ht]
    \centering
    \includegraphics{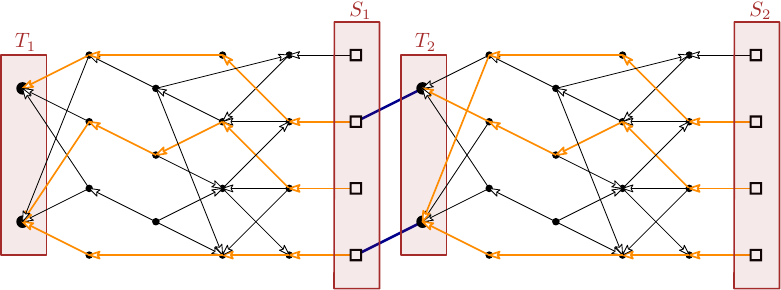}
    \caption{An instance of $2$-level augmentation problem. The two levels are copies of the one level of Figure \ref{fig:augmentation_one_level}, with the same valuation functions for the sinks in each level. The second and fourth source in $S_1$ have a linking edge with the first and second sinks in $T_2$. The set of bold orange edges forms a feasible solution which covers all of $T_1$.}
    \label{fig:augmentation_multi_level}
\end{figure}

\paragraph*{Congestion.}
A technically very useful notion is the following relaxation
of the problem:
we allow each $g_i(e)$ to be an integer number in $\{0,1,\dotsc,\beta\}$ instead of $\{0, 1\}$.
The rest of the definition remains the same.
We say that such a solution has congestion $\beta$.

We will reduce the Submodular Santa Claus problem
to the following gap problem:
for some $T^*\subseteq T_1$ either find a feasible solution $g_1, g_2, \dotsc,g_h$
with coverage $1/\alpha$ and congestion at most $\beta$
or determine that there is no such solution
with coverage $1$ and congestion $1$.
We call an algorithm that solves this problem
an \emph{$(\alpha,\beta)$-approximation algorithm}.
The lower the values of $\alpha$ and $\beta$ are, the
better the approximation rate for the Submodular Santa Claus problem.
The reduction to the augmentation problem follows a very similar
strategy to Chakrabarty, Chuzhoy, and Khanna~\cite{DBLP:conf/focs/ChakrabartyCK09}.
Formally, we prove in \cref{sec:Reduction} the following theorem.
\begin{theorem}\label{th:aug-main}
    Let $\mathcal A$ be an $(\alpha, \beta)$-approximation algorithm for the augmentation problem
    and let $h, \gamma\in\mathbb N$ with $\gamma \ge 1000 \alpha^3\beta^3h^4\log^2(n)$
    and $h\ge 1 + \log(\beta n^2) / \log(\gamma/(2\alpha))$.
    Then there is a $\gamma$-approximation algorithm for the Submodular Santa Claus problem that uses polynomially many calls to $\mathcal A$ on $h$-level instances and has polynomial time overhead.
\end{theorem}
The main technical novelty
is in proving that the augmentation problem can indeed be approximated well.
\begin{theorem}\label{thm:main-tech}
    There is an $(\alpha, \beta)$-approximation algorithm for
    the $h$-level routing problem with running time $n^{O(h)}$ with
    $\alpha = O(1)$ and $\beta = \mathrm{polylog}(n)$.
\end{theorem}
These two theorems imply the main theorem.
\begin{proof}[Proof of \Cref{thm:main}]
    Let $\alpha, \beta$ as in \Cref{thm:main-tech}.
    By setting $h = \lceil 1 + \log(\beta n^2) /\log\log(n) \rceil = O(\log n/\log\log n)$ and $\gamma = \lceil 1000 \alpha^3 \beta^3 h^4\log^2(n) \rceil = \mathrm{polylog}(n)$ we obtain a polylogarithmic approximation in time $n^{O(\log n / \log\log n)}$.

    On the other hand, for $\gamma = n^\epsilon$ and $h = \lceil 1 + \log(\beta n^2) / \log(\gamma/(2\alpha)) \rceil = O(1/\epsilon)$ we have for sufficiently large $n$
    that $\gamma \ge 1000 \alpha^3 \beta^3 h^4 \log^2(n)$. Hence,
    there is a $n^\epsilon$-approximation algorithm with running time $n^{O(1/\epsilon)}$.
\end{proof}

\subsection{Definition of the linear programming relaxation}
\label{subsec:intro-linearprogram}
Consider a solution to the augmentation problem without congestion. This solution exhibits the following hierarchical
structure: after an arbitrary path decomposition 
of the flows $g_1,\dotsc,g_h$, we can associate with each
sink $u\in T_i$ the set of sources $A\subseteq S_i$, for which
in the decomposition some path ends in $u$ and starts in $A$.
Moreover, we associate $L_i(A) \subseteq T_{i+1}$ with $u$
and by the fact that each source has only one outgoing
edge, effectively enforcing a vertex capacity of $1$ on
it, each sink in $T_{i+1}$ 
is only associated to one sink in $T_i$.
Recursively, this results in a forest-like structure of sinks.
Based on this structure, we design our linear programming formulation recursively.

For each suffix of the levels $i,i+1,\dotsc,h$. 
We define $E_{\ge i} = E_i \cup \cdots \cup E_h$.
For each level $i$, set of sinks 
$T^*\subseteq T_i$, and $\alpha, \beta \ge 1$,
we define the linear program $\CLP_{\ge i}(T^*, \alpha, \beta)$.
Here, $\alpha$ is a parameter that stands for the coverage requirement, i.e.,
that every covered sink receives value at least $1/\alpha$, and
$\beta$ stands for the allowed congestion.

In order to model the submodular function requirement, 
we make use of configuration sets $\C(v, \alpha, \beta)$
for each sink $v\in T^*$. These configurations
are the integral flows $g$ in $G_i$ starting in the sources $S_i$ and ending in $v$,
such that the congestion of $g$ is at most $\beta$ and the coverage of $v$,
that is, $f_v(g\cap \delta(v))$, is at least $1/\alpha$.
Each sink in $T^*$ needs to pick one configuration subject to
constraints that will be explained below. The sum of all configurations
stands for the flow $g_i$ in level $i$.

The linear program $\CLP_{\ge i}(T^*, \alpha, \beta)$ contains the variable sets
%$b\in [0,\beta]^{E_{\ge i}}$,
$b_{v,g} \in \mathbb R_{\ge 0}^{E_{\ge i}}$ and $x_{v,g}\in \mathbb R_{\ge 0}$ for all
$v\in T^*, g\in \C(v,\alpha,\beta)$, as well as parameter (fixed constant)
$b\in [0,\beta]^{E_{\ge i}}$.
Here, $b(e)$ describes a ''budget'' of how much flow is allowed to pass through edge $e$,
i.e., an upper bound on $g_i(e)$ where $E_i\ni e$. Including $b$ is necessary for the recursive definition.
The budget is decomposed further into $b_{v,g}(e)$, which describes how much of $b(e)$
is used by the sink $v\in T^*$ together with configuration $g$, using the intuition of a fixed path decomposition as before, which allows
us to trace each unit of flow back to one of the sinks in $T^*$.
We write $(b, b_{v,g}, x_{v,g})\in \CLP_{\ge i}(T^*, \alpha, \beta)$ if these variables and parameters are feasible.
Let $\BLP_{\ge i}(T^*, \alpha, \beta)$ be the set of feasible values $b\in [0,\beta]^{E_{\ge i}}$ for $\CLP_{\ge i}(T^*, \alpha, \beta)$, i.e.,
$b \in \BLP_{\ge i}(T^*, \alpha, \beta)$ if and only if there exist $b_{v,g}, x_{v,g}$
such that $(b, b_{v,g}, x_{v,g})\in \CLP_{\ge i}(T^*, \alpha, \beta)$.
$\BLP_{\ge i}(T^*, \alpha, \beta)$ forms a polytope, which can be seen from
turning $b$ into variables in $\CLP$ and then projecting to $b$.
We are now ready to state the linear program completely.

\begin{LPBox}[Multi-level configuration LP, $\CLP_{\ge i}(T^*, \alpha, \beta)$]
\begin{align}
        \label{eqn:CLP_paths_01}
        \sum_{g \in \C(v, \alpha, \beta)} x_{v, g} &\ge 1 && \forall v\in T^* \\ 
        \label{eqn:CLP_paths_02}
        \sum_{v\in T^*} \sum_{g \in \C(v, \alpha, \beta)} b_{v,g}(e) &\le b(e) && \forall e \in E_{\ge i}\\
        \label{eqn:CLP_paths_03}
        g(e) \cdot x_{v, g} &= b_{v,g}(e) && \forall v\in T^*, g\in \C(v, \alpha, \beta), \\
        \notag
        & && e \in E_i,\\
        \label{eqn:CLP_paths_04}
        (b_{v,g}(e))_{e\in E_{\ge i+1}} &\in x_{v, g} \cdot \BLP_{\ge i+1}(L_i(g\cap S_i), \alpha, \beta) && \forall v\in T^*, g\in \C(v, \alpha, \beta) \\
        \label{eqn:CLP_paths_05}
        x_{v, g} &\ge 0 && \forall v\in T^*, g\in \C(v, \alpha, \beta) \\ 
        b_{v, g}(e) &\ge 0 && \forall v\in T^*, g\in \C(v, \alpha, \beta), \\
        \notag
        & && e\in E_{\ge i}
\end{align}
\end{LPBox}
For the last level $i = h$ we omit Constraint~\eqref{eqn:CLP_paths_04}.

Constraint~\eqref{eqn:CLP_paths_01} ensures that each sink in $T^*$ selects one configuration. Constraint~\eqref{eqn:CLP_paths_02} enforces the relationship between $b(e)$ and $b_{v,g}(e)$.
Constraint~\eqref{eqn:CLP_paths_03} guarantees that $b_{v,g}(e)$ correctly represents
the amount of flow on edge $e$ caused by $v$ and $g$ in level $i$.

The last Constraint~\eqref{eqn:CLP_paths_04} requires some more elaboration. First, we verify that it is indeed a polyhedral constraint:
for some $v, g$, the values $x_{v,g}$ and $b_{v,g}(e)$, $e\in E_{\ge i+1}$, that satisfy~\eqref{eqn:CLP_paths_04}
are exactly those generated by the polyhedral cone
with extreme rays $x_{v,g} = 1$ and 
$b_{v,g}(e)$, $e\in E_{\ge i+1}$, being a vertex of
$\BLP_{\ge i+1}(L_i(g\cap S_i), \alpha, \beta)$.

The intuition of the constraint is that if $v$ is covered
via the flow $g$, then $L_i(g\cap S_i)$ need to be covered
in the next level. However, it is not sufficient that
$\CLP_{\ge i+1}(L_i(g\cap S_i), \alpha, \beta)$ is feasible,
since several sinks in $T_i$ (not just $v$) share the budget on edges in $E_{\ge i+1}$. Hence, we use $b_{v,g}(e)$
to store the budget used by $v$ (together with configuration $g$). Constraint~\ref{eqn:CLP_paths_02} then ensures that
the flow used by all sinks together does not exceed
the budget.
In the uncongested case and with the intuition of the forest-like structure,
this simply says that the trees rooted in different sinks of $T^*$
are edge-disjoint.
Formally, the fact that we can separately consider solutions
induced by the different sinks $L_i(g\cap S_i)$ in the next level
is justified by the following lemma.
\begin{lemma}\label{lem:separate}
    Let $T^*, T^{**}$ be disjoint sets of sinks
    and let $b\in [0, \beta]^{E_{\ge i}}$.
    Then
    $b \in \BLP_{\ge i}(T^* \cup T^{**}, \alpha, \beta)$
    if and only if there exist $b' + b'' = b$ with
    $b'\in \BLP_{\ge i}(T^*, \alpha, \beta)$ and $b''\in \BLP_{\ge i}(T^{**}, \alpha, \beta)$.
\end{lemma}
This means that for an integral uncongested solution Constraint~\eqref{eqn:CLP_paths_04} is equivalent to
\begin{equation*}
    (b(e))_{e\in E_{\ge i+1}} \in B_{\ge i+1}(L_i(g_i\cap S_i), 1, 1) \ ,
\end{equation*}
where $g_i = \sum_{v\in T^*} \sum_{g\in\C(v,\alpha,\beta)} x_{v,g} g$.
The implication uses the fact that $L_i(g'\cap S_i)$ and $L_i(g''\cap S_i)$ must be disjoint for different $g', g''$ used by the solution: 
$g' \cap S_i$ and $g'' \cap S_i$ are disjoint since
each vertex in $S_i$ has out-degree $1$, enforcing a vertex capacity of $1$ on $S_i$ and $L_i(\cdot)$ is injective.
The proof of the lemma is straight-forward and deferred to Appendix~\ref{app:separate}.

\section{Reduction to the augmentation problem}
\label{sec:Reduction}
We now present the reduction of the Submodular Santa Claus problem to the augmentation problem that we have introduced in Section~\ref{subsec:intro-reduction}.

In order to devise a $\gamma$-approximation algorithm,
by a standard binary search framework it suffices for a given $\eta$ to
either find a solution of value at least $\eta / \gamma$ or to determine
that $\OPT < \eta$. Furthermore, by scaling all functions $f_p$ we may assume that $\eta = 1$.

\subsection{From general instances to canonical instances.}

We will first reduce to the following \emph{canonical instances}.

\paragraph*{Canonical instance.} As mentioned above we need to either determine that $\OPT < 1$ or find a solution of value at least $1/\gamma$.
We distinguish between \emph{basic players} $B$ and \emph{complex players} $C$. 

For a basic player $p\in B$, we have that $f_p(S) \in \{0, 1\}$ for
all $S\subseteq R$. Notice that $f_p(S) = 1$ if and only if $S$ contains a resource $r$ with $f_p(\{r\}) = 1$. We may therefore assume without loss of generality that
each basic player gets exactly one resource of value $1$ in a solution. 

Each complex player $p$ has a private resource $r(p)$ with $f_{p}(\{r(p)\}) = 1$ and
$f_{q}(\{r(p)\}) = 0$ for all complex players $q\neq p$.
Similar to before, we may assume that if player $p$ gets $r(p)$ in a solution then
$p$ does not get any other resources. For all resources $r\neq r(p)$, we have
$f_p(\{r\}) < 1/\gamma$.

\paragraph*{Reduction to canonical instances.} 
In a general instance, we split each player $p\in P$ into a basic player $p'$ and a complex player $p''$ and we introduce an additional resource $r(p'')$ which has value $f_{p'}(\{r(p'')\}) = f_{p''}(\{r(p'')\}) = 1$ for $p'$ and $p''$ and value $0$ for all other players. For player $p$, we define $R_p^{(b)}$ the set of resources $r$ such that $f_{p}(\{r\}) \ge 1/\gamma$ (i.e. the big resources for $p$). Then, we can define the submodular valuation function for player $p'$ as follows.
\begin{equation*}
    f_{p'}(S):=\begin{cases}
        1 \mbox{ if $S\cap (R_p^{(b)}\cup \{r(p'')\})\neq \emptyset$,}\\
        0 \mbox{ otherwise.}
    \end{cases}
\end{equation*}
For player $p''$, we define the submodular valuation function as follows
\begin{equation*}
    f_{p''}(S):=
    \begin{cases}
        1 \mbox{ if $r(p'')\in S$,}\\
        f_p(S\setminus R_p^{(b)}) \mbox{ otherwise.}
    \end{cases}
\end{equation*}

Note that the resource $r(p'')$ can cover either $p'$ or $p''$. This corresponds intuitively
to the fact that $p$ needs either small resources that sum up to a large value or a single resource of sufficiently large value. Note that in the above construction $p'$ is a basic player which values only the additional resource $r(p'')$ or the big resources of $p$, while $p''$ is a complex player which values only the additional resource or the small resources of $p$. 

It is easy to see that a solution of value at least $1$ in the original instance can be transformed to a solution of value at least $1$ in the canonical instance: if a player $p$ receives at least one big resource $r$ in the orginal instance such that $f_p(r)\ge 1/\gamma$, we give that same resource to the corresponding basic player $p'$ in the canonical instance, and the complex player $p''$ is given the additional resource $r(p'')$. On the other hand, if a player $p$ does not receive any big resource, we give the resource $r(p'')$ to the basic player $p'$ in the canonical instance, and the complex player $p''$ receives the same resources as $p$ receives in the original instance.

Conversely, it is easy to see that a solution of value at least $1/\gamma$ in the canonical instance
can be transformed to a solution of value at least $1/\gamma$ in the original instance. Indeed, if a pair of players $p',p''$ (corresponding to one player $p$ in the original instance) both receive value at least $1/\gamma$, it must be that either (a)
$p'$ does not receive the resource $r(p'')$ hence $p'$ must receive one resource of value at least $1/\gamma$ for $p$ in the original instance, or (b) the player $p''$ does not receive the resource $r(p'')$ hence must receive a bundle of resources $S$ of total value at least $1/\gamma$ for player $p$ in the original instance. In both cases, the original player $p$ is covered.

Hence, it suffices to devise an algorithm for the canonical instance.

\subsection{From canonical instances to augmentation problem}
Consider a partial assignment of resources $\sigma:R \rightarrow P\cup \{\bot\}$, where symbol $\bot$ is used to describe that a resource is not assigned to any player.
From a canonical instance, partial assignment $\sigma$, and
a parameter $h\in\mathbb N$,
we will construct an instance $I(\sigma, h)$ of the augmentation problem that
closely relates to potential reassignments of resources.
Parameter $h$, the number of levels in the instance, will influence the approximation ratio and the running time.
First, we define
\begin{equation*}
   \overline \sigma(r)= 
   \begin{cases} 
   \sigma(r) \mbox{ if } \sigma(r)\neq \bot \mbox{ and } f_{\sigma(r)}(\{r\})=1\ ,\\
    \bot \mbox{ otherwise.}
    \end{cases}
\end{equation*}
Intuitively, the assignment $\overline \sigma$ is equal to the assignment $\sigma$ except that all complex players release the small resources assigned to them. Notice that in this new assignment $\overline \sigma$, a complex player $p$ can only receives its unique big resource $r(p)$, if any.
The intuition here is that this greatly simplifies the structure of potential reassignments, since every player can now only give up a single resource (i.e., has an out-degree of at most $1$). Thus, reassignments can be thought of as directed trees.

\paragraph*{Construction of augmentation instance.}
All our $h$ levels will feature the same graph, which is defined as follows.
Let $G = (V, E)$ where $V$ consists of all resources $R$,
all basic players $B$, two copies of the complex players, which we denote by $C^S$ and $C^T$,
a vertex $t$ and a vertex $s(r)$ for each currently unassigned resource $r$ (i.e., $\overline\sigma(r) = \bot$).
For each resource $r$ and each $q\in B\cup C^T$ for which $f_q(\{r\}) > 0$
there is an edge from each $r$ to $q$ if $\overline{\sigma}(r)\neq q$. Further, for each $q\in B\cup C^S$
and each resource $r$ which is currently assigned to $q$ (in the assignment $\overline \sigma$)
there is an edge from $q$ to $r$.
Recall that by definition of $\overline \sigma$ this implies that $f_q(\{r\}) = 1$.
Finally, there is an edge from $s(r)$ to $r$ for each unassigned resource $r$ and an edge from each basic player that is currently not assigned any
resource to $t$ (again referring to the assignment $\overline \sigma$).

The sinks are $t$ and the copies $C^T$, and the sources are $s(r)$ for unassigned resources $r$ and the
copies $C^S$.
The incoming edges for some $p\in C^T$ all come from resources and thus, a set
$A\subseteq \delta(p)$ naturally corresponds to the set of resources incident
to it. We define $f_p(A)$ as the function value for these resources and complex player
$p$ in the canonical instance.

For vertex $t$ we define $f_t(A) = |A| / |\delta(t)|$ for each $A\subseteq \delta(t)$. Notice that this function is linear, hence submodular.

In the reassignment of resources corresponding to the optimal
solution, any complex player $p$ whose private resource $r(p)$ is taken away would receive
a lot of other resources.
The intuition for $C^T$ and $C^S$ is that we do not strictly enforce this:
from the copies in $C^S$ we potentially take away the
private resource and the copies in $C^T$ potentially receive a lot of other resources,
but consistency is not enforced.
This relaxation is to make it easier to find a good reassignment. However, we still have to strengthen this relaxation to avoid that a reassignment
simply takes away all private resources from $C^S$ without being able to cover them
with other resources.

Towards this, we build a multi-level instance of the augmentation problem by stacking several copies of the graph defined as above
on top of each other. Then we connect these graphs together by some linking edges. Formally, for each $1\le i< h$ and every complex player $p$, we add a linking edge from the vertex in $C^T$ corresponding to player $p$ in level $i+1$ to the vertex in $C^S$ corresponding to the player $p$ in level $i$.

By construction, this enforces that if the unique resources of some complex
players $A\subseteq C^S$ of level $i$ are removed, then in the next level there must be a solution
that gives a lot of resources to the elements in $C^T$ of level $i+1$ that correspond to $A$.

We denote by $I(\sigma, h)$ the multi-level augmentation instance obtained as above.
\paragraph*{Existence of an augmentation.}
In the following lemma we prove that the instance $I(\sigma,h)$
as above is feasible, assuming that the canonical instance is.
\begin{lemma}\label{lem:existence_of_flow}
Consider the instance $I(\sigma, h)$ of the augmentation problem, where $\sigma$ is an arbitrary assignment of resources. If there exists a solution of value $1$ for the canonical instance, then there exists a solution with coverage $\alpha = 1$ and congestion $\beta = 1$ for $I(\sigma,h)$ for any $h\ge 1$, which covers the sink
$t$ in level $1$.
\end{lemma}
\begin{proof}
    As a solution, we define the same flow in each level. Let $\sigma_\OPT$ be the optimal assignment in the canonical instance, and $\overline \sigma$ the modified assignment derived from $\sigma$ as before. We assume that there is no resource $r$ such that $\sigma_{\OPT}=\bot$. This is without loss of generality, because we can always assign this resource to an arbitrary player and modify $\sigma_{\OPT}$ accordingly.
    
    For simplicity of notation, we denote by $(p, r)$ the edge from the vertex corresponding to player $p$ to the resource $r$. Notice that if $p$ is a complex player, which means that there are two vertices corresponding to $p$, this edge only exists from the copy of $p$ in $C^S$. Thus, the edge is uniquely defined. Similarly, we denote by $(r, p)$ the edge from a resource $r$ to a player $p$. Again, if $p$ is a complex player, the edge must go to $C^T$. We also have edges $(s(r),r)$ for unassigned resource $r$ and edges $(p, t)$ for uncovered players to $t$.

    The solution flow $g_i(e)$ for each level $i$ and some edge $e$ is defined by
\begin{equation*}
g_i(e) = 
    \begin{cases}
        1 \quad \mbox{ if $e=(p,r)$ and $\overline{\sigma}(r)\neq \sigma_\OPT(r)$},\\
        1 \quad \mbox{ if $e=(r,p)$ and $\sigma_\OPT(r)=p$},\\
        1 \quad \mbox{ if $e=(s(r),r)$},\\
        1 \quad \mbox{ if $e=(p,t)$},\\
        0 \quad \mbox{ otherwise.}
    \end{cases}
\end{equation*}
It is easy to verify the validity of this solution since our flow solution mimics the optimal assignment. One can verify that the sink $t$ in level $1$ is covered since all uncovered basic players send a flow of $1$ to $t$, and the congestion is clearly at most $1$ on any edge.

Second, any player vertex $p$ which is not a source nor a sink must be a basic player. Therefore, if the corresponding vertex is traversed by some flow, there is one unit of outgoing flow since $p$ is assigned at most one big resource in $\overline \sigma$, and exactly one unit of in-going flow since $p$ receives only one big resource in $\sigma_\OPT$. Hence we have the flow conservation at all player vertices. The case of resource vertices is very similar, the resource is assigned to exactly one player in $\sigma_\OPT$ and at most one in $\overline \sigma$, and it is easy to verify that flow conservation holds at the corresponding vertex. Either the resource is not traversed by any flow if both $\overline \sigma$ and $\sigma_\OPT$ agree on that resource, or it is traversed by exactly one flow unit.

Finally, in the flow solution $g_i$, the set of vertices in $C^S$ which send some flow corresponds to a set of complex players which give up their big resource in the reassignment. But since $\sigma_{\OPT}$ covers all players, it must be that those players are covered by some small resources, hence the corresponding sinks in $C^T$ will receive enough flow in the assignment $g_{i+1}$, which satisfies the constraint related to the linking edges between $C^S$ in one layer and $C^T$ in the next layer.
\end{proof}

\paragraph*{Approximate solution with additional structure.}
By \Cref{lem:existence_of_flow} we know that $I(\sigma, h)$
will have a feasible solution, assuming the canonical instance
is feasible. We will show next that by a negligible loss
we can simplify any solution to obey a structure that
will later help in actually augmenting the assignment $\sigma$.

\begin{lemma}\label{lem:existence_of_structured_flow}
Let $\alpha, \beta \in\mathbb N$ and $h \ge 1 + \log(\beta n^2) / \log(\gamma / (2\alpha))$ and consider the instance $I(\sigma, h)$ constructed from the assignment $\sigma$.

Any solution $g_1,g_2,\dotsc,g_h$ of coverage $\alpha$ and congestion $\beta$ that covers $t$ in level $1$
can be turned in polynomial time into a solution $g'_1,g'_2,\dotsc,g'_h$ with coverage $2h\alpha$ and congestion $\beta$
such that
\begin{enumerate}
    \item either (a) $g_1$ uses none of the source in $C^S$ (only the sources $s(r)$)
    or (b) $g_1$ does not use sources $s(r)$ (and therefore only $C^S$); and
    \item $g_h$ does not use any sources in $C^S$ (only potentially sources $s(r)$)\ .
\end{enumerate}
\end{lemma}
\begin{proof}
    Let $i$ be a level and consider a path decomposition of
    $g_i$ into paths that each send unit flows.
    We define weights for each path corresponding
    to the marginal values. Specifically, for a sink $v$
    and an arbitrary ordering of the paths $P_1,P_2,\dotsc$
    that end in $v$, set
    \begin{align*}
        w(i, P_j) &:=f_p(\delta(v) \cap P_j \mid \delta(v) \cap \{P_1,P_2,\ldots ,P_{j-1}\}) \\
        &= f_p(\delta(v) \cap \{P_1,P_2,\ldots ,P_{j}\})-f_p(\delta(v) \cap \{P_1,P_2,\ldots ,P_{j-1}\})\ .
    \end{align*}
    The total weight of paths ending in a covered sink is at least $1/\alpha$.
    
    For the construction we need the notion of a subtree:
    with each path $P$ (in some level $i$) we can associate the following subtree of paths: if $P$ starts in some source $s(r)$, then the subtree only contains $P$. If it starts in some vertex of $C^S$ corresponding to a complex player $p$, then we associate with $P$ all paths in level $i+1$ that end in the corresponding vertex of $C^T$ and their recursively defined subtrees.
    
    Now, if the paths to sink $t$ in level~$1$ that start in one of the sources $s(r)$ have weight at least $|\delta(t_1)|/(2\alpha)$, then we simply keep these paths and delete all the others as well as all flows in later levels. The obtained flow loses only a factor of $2$ in the coverage at $t$ in level $1$ due to submodularity and satisfies the conditions of (1a) and (2).
   
    Otherwise, it must be that paths of level $1$ that start in $C^S$ and end in $t$ have a total weight of at least $|\delta(t_1)|/(2\alpha)$. We drop all other paths including their subtrees, thus satisfying (1b).
    Next, we proceed to establish Property~(2).
    We assume without loss of generality that the solution is minimal in the sense that it contains exactly the subtree of $t$ in level $1$ and no other paths.
    
    We mark the covered complex player vertices according to their ``depth''. Every vertex in $C^T$ of some level $i\ge 1$, which receives more than a $1/(2h)$ fraction of its weight through paths from a source $s(r)$ in the same level, is marked as \emph{depth-1} vertex, and we delete all other paths that end in them (along with their subtree).
    The corresponding vertices of $C^S$ in level $i-1$ (linked to the marked one via linking edges) are considered to have the same depth of $1$.
    
    Then we proceed iteratively. Having marked all vertices up to depth $\ell$ for some $\ell\ge 1$, we say that every unmarked vertex in $C^T$ of some level $i\ge 1$, which receives more than a $1/(2h)$ fraction of its weight via paths from depth-$\ell$ vertices in $C^S$ of the same level are marked as depth-$(\ell+1)$ vertices (together the corresponding vertices in $C^S$ of level $i-1$), and we delete all other paths that end in them along with the corresponding subtree.

    We claim that if we choose $h$ as in the lemma
    then all covered complex players of level $2$ are marked by iteration $\ell=h-1$. Assume otherwise. Let $p\in C^T$ be the complex player of level $2$, which is not marked. Since this player is unmarked, it receives less than a $1/(2h)$ fraction of his flow from the sources $s(r)$ or level-$\ell'$ players, for any $1\le \ell' < h$. 
    Furthermore, there cannot be depth-$h$ players in level $2$, since the remaining levels are only $h-1$.
    Hence, it must be that at least a $1/2$ fraction of his weight comes from unmarked players.
    Applying the same argument recursively, the player $p$ in level $2$ must be the root of a tree of depth $h-1$ with minimum out-degree at least 
    $\gamma/(2\alpha)$,
    since by construction every resource has marginal value at most $1/\gamma$.

    It follows that the number of paths in level $h$ is at least
    \begin{equation*}
        (\gamma/(2\alpha))^{h-1} > \beta n^2 \ ,
    \end{equation*}
    which is a contradiction since then some edge would have congestion more than $\beta$.
    
    Notice that at the end of this process, none
    of the sources in $C^S$ are used, since the leafs of subtrees induced by players $C^T$ of level $2$
    start in $s(r)$ of some level. In this process, we lost an approximation factor of at most $2h$, which concludes the proof.
\end{proof}

\subsection{From approximate augmentation to canonical instance solution}
Our approach is to start with a solution that covers all complex players $p$
with their private resource $r(p)$ and we iteratively reduce the
number of basic players that are not covered
while maintaining a good coverage of the complex players. During this process, it will be helpful to maintain an assignment $\sigma_{k}:R\rightarrow P$ of resources to players, for every iteration $k = 1,2,\dotsc$ 
\begin{lemma}[Augmentation]\label{lem:aug}
Let $\alpha, \beta, k \in\mathbb N$ and
$h \ge 1 + \log(\beta n^2) / \log(\gamma / (2\alpha))$.
Assume that there exists a solution of value $1$ for the canonical instance with parameter $\gamma$. Given an assignment $\sigma_k$ for this canonical instance where each complex player 
    gets a value of at least $1/(8\alpha \beta h^2 k) - 4k/\gamma$, and a solution to the augmentation problem $I(\sigma_k, h)$ 
    with coverage $\alpha$ and congestion $\beta$,
    one can find in polynomial time an assignment $\sigma_{k+1}$ where each complex player 
    gets a value of at least $1/(8\alpha \beta h^2 (k+1)) - 4(k+1)/\gamma$
    and the number of basic players not covered reduces by a factor of $(1 - 1/(4h\alpha\beta))$.
\end{lemma}
\begin{proof}
We first take 
away all resources $r \neq r(q)$ from each complex player $q$ to obtain the assignment $\overline{\sigma}_k$, as in the construction of the augmentation instance. With \Cref{lem:existence_of_structured_flow}, we transform the solution to the augmentation problem for the assignment $\overline \sigma_k$ into an augmentation solution $g_1,\dotsc,g_h$ covering $t$ in level $1$
that has
coverage $2h\alpha$, congestion $\beta$, and the structural assumptions
mentioned. In particular, the solution is of one of two types, for which
we derive the assertion separately.

\paragraph*{Flow $g_1$ does not use $C^S$.}
Assume that $g_1$ uses only sources $s(r)$. This means
$g_1$ forms a flow of at least $|\delta(t)|/(2h\alpha)$ from sources $s(r)$ to $t$ with
congestion $\beta$. Consider the fractional flow $g_1/\beta$. This flow has
congestion at most $1$ and flow value at least $|\delta(t)|/(2h\alpha\beta)$. By standard flow
arguments we can then also find an integral flow $g$ from sources $s(r)$ to $t$ with
congestion $1$ and flow value at least $|\delta(t)|/(2h\alpha\beta)$.
This flow can be interpreted as a reassignment of resources.
Let us denote by $\sigma_{k}'$ the assignment obtained from $\overline \sigma_k$ following the reassignment.
$\sigma_{k}'$ covers a $1/(2h\alpha\beta)$-fraction of previously uncovered
basic players and each basic player covered in $\overline \sigma_{k}$ remains covered.

However, it might be that some complex players covered in $\sigma_k$ by small resources are not covered anymore in $\overline \sigma_k$, hence not covered either in $\sigma_{k}'$. Consider such a complex player $p$. We have two cases. 

If at least two of the small resources that are assigned to $p$ in $\sigma_k$ are taken away by some other player in $\sigma_{k}'$, then we give $p$ back the private resource $r(p)$ and modify the assignment $\sigma_{k}'$ accordingly. If $r(p)$ is taken by some basic player in $\sigma_{k}'$, this results in
the basic player being uncovered.

Otherwise, we modify $\sigma_{k}'$ by giving to $p$ all the resources it was assigned in $\sigma_k$, except for the one resource $r$ (if any) that is already used in $\sigma_{k}'$. Notice that in that case, $p$ receives value at least
\begin{equation*}
    f_p(R' \setminus\{r\}) \ge f_p(R') - f_p(\{r\}) \ge \frac{1}{8\alpha k \beta h^2} - \frac{4k}{\gamma} - \frac{1}{\gamma} \ge \frac{1}{8\alpha (k + 1) \beta h^2} - \frac{4(k+1)}{\gamma} \ ,
\end{equation*}
where $R'$ are the resources that were assigned to $p$ in $\sigma_k$.
We let $\sigma_{k+1}$ be the assignment resulting from these
modifications.

To conclude the analysis of this case, let $N$ be the number of complex players which take back their private resource in the first case above. For each one of these players, we uncover one basic player, but each of these players also sends $2$ flow units to $t$ via different paths. Hence, the change in the number of covered basic players is at least
\begin{equation*}
    \max\{2N - N, \frac{|\delta(t)|}{2 h\alpha \beta} - N\} \ge \frac{|\delta(t)|}{4 h \alpha \beta}\ ,
\end{equation*}
where $|\delta(t)|$ is the number of previously uncovered basic players.

\paragraph*{Flow $g_1$ only uses $C^S$.} 
Define $g = g_1 + g_2 + \cdots + g_h$.
Let $X\subseteq C$ be the complex players for which the corresponding vertex in $C^S$
has any outgoing flow in $g$ and let $Y\subseteq C\setminus X$ be the complex players $p$ that are not assigned their private resource $r(p)$ in $\sigma_k$.
By \Cref{lem:existence_of_structured_flow}, the flow $g$ has the following properties:
\begin{enumerate}
    \item $g(e) \le \beta h$ for all $e\in E$,
    \item the incoming flow to $t$ is at least $|\delta(t)| / (2h\alpha)$,
    \item For each $p\in X$ the corresponding copy $p^T\in C^T$ satisfies $f_{p^T}(\delta(p^T) \cap g) \ge 1/(2h\alpha)$.
\end{enumerate}
The last property holds because $C^S$ in the last level is not used.
For $p\in X$ and the corresponding copy $p^T\in C^T$
write $E(p) = g\cap \delta(p^T)$ and
for $p\in Y$ we define $E(p)$ as the edges from $s(r)$ to $r$ for resources $r$ assigned to $p$ in $\sigma_{k}$ (those that were removed from $p$ in $\overline\sigma_k$).

For some $p\in X\cup Y$ let $\{e_1,e_2,\dotsc,e_{\ell}\} = E(p)$ be an arbitrary but fixed order of the set $E(p)$.
We define a partition
$E(p) = E_1(p) \cup E_2(p) \cup \dotsc$ by marginal value, with $E_i(p)$
consisting of all $e_j$ with
\begin{equation*}
2^{-(i-1)} > f_{p^T}(\{e_1,\dotsc,e_j\}) - f_{p^T}(\{e_1,\dotsc,e_{j-1}\}) \ge 2^{-i} \ .
\end{equation*}
Notice that $E_1(p), E_2(p),\dotsc, E_{\lfloor\log \gamma \rfloor}(p)$ are actually empty,
since $f_p(\{e\})\le \gamma$ for all $e\in E(p)$.
We can now find an integer flow $g'$
with
\begin{enumerate}
    \item $g'(e) \in \{0, 1\}$ for all $e\in E$
    \item the incoming flow to $t$ is at least $\frac{|\delta(t)|}{4h\alpha\beta}$
    \item For each $p\in X\cup Y$ and $i = 1,2,\dotsc$ we have
    \begin{equation*}
    \lfloor \frac{g(E_i(p))}{2(k+1)\beta h} \rfloor \le g'(E_i(p)) \le \lceil \frac{g(E_i(p))}{2(k+1)\beta h} \rceil \ .
    \end{equation*}
\end{enumerate}
This holds because the fractional flow $g_1 / 2 + g / (2(k+1)\beta h)$ satisfies
Properties~2 and~3 and has congestion at most $1$: here notice that $g_1$ has no
flow on any of the edge sets $E_i(p)$.
Arguing with dummy vertices for each set $E_i(p)$,
it follows easily by standard flow arguments that this fractional flow is a convex
combination of flows that satisfy~1 and~3, at least one out of
which then also satisfies~2.

We use this flow $g'$ to transform $\sigma_k$ into a new assignment $\sigma_{k+1}$ in the natural way: From the flow $g'$ we can transform $\overline \sigma_k$ (itself obtained from $\sigma_k$) into a new assignment $\sigma_k'$. Then, we modify $\sigma_k'$ by giving to each player $p\in Y$ all the resources that $p$ was assigned in $\sigma_k$ and that are not traversed by any flow in $g'$. This constitutes our final assignment $\sigma_{k+1}$. We will analyze that the new assignment satisfies the properties in \Cref{lem:aug}.

First, each basic player that loses its current resource gets a new
resource, because basic players are not sources. This means that
all basic players that had a resource in the current assignment, still
have one in the new assignment. Furthermore, for each unit of incoming flow to $t$,
there must be one basic player that previously did not have a resource and
now gets one. This means that the number of basic players that do not
have a resource decreases by a factor of $(1 - 1/(4h \alpha \beta))$.

Consider now a complex player $p\in X$, i.e., $p$
may lose its private resource $r(p)$. Then the value of the resources assigned
to $p$ through the reassignment is
\begin{align*}
     f_p(E(p) \cap g')
      &\ge \sum_{i = 1}^{\infty} 2^{-i} \cdot g'(E_i(p)) \\
        &\ge \sum_{i = 1}^{\infty} 2^{-i} \cdot \lfloor \frac{g(E_i(p))}{2(k+1) \beta h} \rfloor \\
        &\ge \frac{1}{4 (k+1) \beta h} \sum_{i = 1}^{\infty} 2^{-(i - 1)} \cdot g(E_i(p)) - \sum_{i = \lfloor \log_2(\gamma) \rfloor}^{\infty} 2^{-i} \\
        &\ge \frac{1}{4 (k+1) \beta h} f_p(E(p) \cap g) - \frac{4}{\gamma} \\
        &\ge \frac{1}{4 (k+1) \beta h} \cdot \frac{1}{2h\alpha} - \frac{4}{\gamma} \\
        &\ge \frac{1}{8\alpha \beta h^2 (k+1)} - \frac{4(k+1)}{\gamma} \ . \\
\end{align*}
Finally, consider a complex player $p\in Y$, i.e., $p$ is not assigned $r(p)$ in
the assignment $\sigma_k$. The reassignment may further take away resources,
but we will argue that $p$ retains a large value. More precisely, the
value of resources that $p$ still has after the reassignment is
\begin{align*}
     f_p(E(p) \setminus E(g'))
      &\ge f_p(E(p)) - \sum_{i = 1}^{\infty} 2^{-(i-1)} \cdot g'(E_i(p)) \\
      &\ge f_p(E(p)) - \sum_{i = 1}^{\infty} 2^{-(i-1)} \cdot \lceil \frac{g(E_i(p))}{2(k+1) \beta h} \rceil \\
      &\ge f_p(E(p)) - \frac{1}{(k+1) \beta h} \sum_{i = 1}^{\infty} 2^{-i} \cdot g(E_i(p)) - \sum_{i = \lceil\log_2(\gamma) \rceil}^{\infty} 2^{-i} \\
      &\ge f_p(E(p)) \left(1 - \frac{1}{(k+1) \beta h}\right) - \frac{4}{\gamma} \\
      &\ge \left(\frac{1}{8\alpha\beta h^2 k} - \frac{4k}{\gamma}\right) \left(1 - \frac{1}{(k+1)}\right) - \frac{4}{\gamma} \\
      &\ge \frac{1}{8\alpha\beta h^2 (k+1)} - \frac{4(k+1)}{\gamma} \ .
\end{align*}
We can conclude that the new assignment $\sigma_{k+1}$ satisfies our desiderata in both cases.
\end{proof}

\begin{proof}[Proof of \Cref{th:aug-main}]
We repeatedly apply \Cref{lem:aug} with $h$ as specified in the lemma.
After at most $k = 4h\alpha\beta\log(n)$ iterations all basic players are covered.
According to the lemma, in the resulting assignment each complex player receives a value of at least
\begin{equation*}
    \frac{1}{32\alpha^2\beta^2 h^3 \log(n)} - \frac{16 h\alpha\beta \log(n)}{\gamma} \ .
\end{equation*}
This is at least $1/\gamma$ provided that
\begin{equation*}
    \gamma \ge 1000 \alpha^3\beta^3h^4\log^2(n) \ . \qedhere
\end{equation*}
\end{proof}

\section{Solving the linear programming relaxation}
\label{sec:linearprogram}
This section is devoted to proving the following theorem.
\begin{theorem}\label{thm:linprog}
    Let $\alpha = 40$ and $\beta = 10\log(n)$.
    There is a Las Vegas algorithm that
    given some $b\in [0, \beta]^{E_{\ge i}}$ determines that $b\in \BLP_{\ge i}(T^*, \alpha, \beta)$ and finds corresponding variables for $\CLP_{\ge i}(T^*,\alpha,\beta)$ or finds a hyperplane
    that separates $b$ from $\BLP_{\ge i}(T^*, 1, 1)$.
    The algorithm makes polynomially many recursive calls on $\BLP_{\ge i+1}$ and has otherwise polynomial time overhead.
\end{theorem}
It follows immediately that one can find a solution to $\CLP_{\ge 1}(T^*, \alpha, \beta)$ for a given $b\in [0,\beta]^E$ and $T^*\subseteq T_1$ in time $n^{O(h)}$
with $\alpha,\beta$ as in the theorem.

\subsection{Reduction to separation problem}
In this subsection we reformulate the linear program using Dantzig-Wolfe decomposition, which then can be solved using the Ellipsoid method. This will reduce the proof of
\Cref{thm:linprog} to a certain separation problem, which we then solve in the next
subsection.
Dantzig-Wolfe decomposition, see~\cite{dantzig1960decomposition},
is a reformulation method
of specific block structured linear programs.
The specific variant we are interested in is given in the following lemma.
\begin{lemma}\label{lem:dantzig-wolfe}
Suppose we are given a linear program with $k$ sets of variables
$x^{(1)} \in \mathbb R^{n_1}, x^{(2)} \in \mathbb R^{n_2}, \dotsc, x^{(k)} \in \mathbb R^{n_k}$.
There are local constraints given by $B_i x^{(i)} \le b^{(i)}$
for some $B_i\in \mathbb R^{m_i \times n_i}$ and $b^{(i)}\in \mathbb R^{m_i}$ as well as global constraints of the form $A_1 x^{(1)} + \cdots + A_k x^{(k)} \le b^{(0)}$ for $A_i\in\mathbb R^{m_0 \times n_i}$ and $b^{(0)}\in\mathbb R^{m_0}$.
We assume that each set
$Q_i = \{x^{(i)}\in\mathbb R^{n_i} : B_i x^{(i)} \le b^{(i)}\}$ is a polyhedral cone generated by a finite set of extreme rays $R_i$.     Then this linear program is feasible if and only if
    there is a solution to
    \begin{align*}
        A_1 \sum_{r\in R_1} \lambda^{(1)}_r r + \cdots + A_1 \sum_{r\in R_k} \lambda^{(k)}_r r &\le b^{(0)} \\
        \lambda^{(i)} &\in \mathbb R_{\ge 0}^{R_i} &\forall i = 1,2,\dotsc,k
    \end{align*}
\end{lemma}
\begin{proof}
    If there is a solution $x^{(1)},\dotsc,x^{(k)}$ to the original linear program, then we can write each $x^{(i)}\in Q_i$ as a conic combination of the extreme rays $R_i$.
    The corresponding weights $\lambda^{(i)}$ form a solution to the reformulation.

    Suppose on the other hand there exists a solution $\lambda^{(1)}, \dotsc, \lambda^{(k)}$ to the reformulation. Then
    $x^{(i)} := \sum_{r\in R_i} \lambda^{(i)} r \in Q_i$ is a solution to the original linear program.
\end{proof}
This reformulation will be useful to us because it can
greatly reduce the number of constraints at the cost of increasing
the number of variables.
To apply it, we first need to make some preparations.
Let $\EX_{\ge i}(T^*, \alpha, \beta) \subseteq \BLP_{\ge i}(T^*, \alpha, \beta)$
be the set of extreme points of the polytope, in particular,
\begin{equation*}
    \mathrm{conv}(\EX_{\ge i}(T^*, \alpha, \beta)) = \BLP_{\ge i}(T^*, \alpha, \beta) \ .
\end{equation*}
Define further
\begin{align*}
Q_{v, g} = \{(x_{v, g}, b_{v, g})\in \mathbb R_{\ge 0}\times \mathbb R_{\ge 0}^{E_{\ge i}} &: (b_{v,g}(e))_{e\in E_{\ge i+1}} \in x_{v,g}\cdot \BLP_{\ge i+1}(L_i(g\cap S_i), \alpha,\beta), \\
&\quad g(e)\cdot x_{v,g} = b_{v,g}(e) \ \forall e\in E_i \} \ .
\end{align*}
Then $Q_{v,g}$ is a polyhedral cone generated
by the extreme rays $(1, g, d)\in \mathbb R_{\ge 0}\times \mathbb R_{\ge 0}^{E_i} \times \mathbb R_{\ge 0}^{E_{\ge i+1}}$ for each $d\in \EX_{\ge i+1}(L_i(g\cap S_i), \alpha, \beta)$.

Given $b\in [0, \beta]^{E_{\ge i}}$ as fixed parameters, we
apply Dantzig-Wolfe decomposition (\Cref{lem:dantzig-wolfe}) to obtain the formulation $\DW$, which is equivalent to $\CLP$
in the sense that one LP is feasible for $b$ if and only if the other is.

\begin{LPBox}[Dantzig-Wolfe decomposition of multi-level configuration LP, $\DW_{\ge i}(T^*, \alpha, \beta)$]
\begin{align}
        \label{eqn:CLPDW_paths_01}
        \sum_{g \in \C(v, \alpha, \beta)} \sum_{d\in \EX_{\ge i+1}(L_i(g\cap S_i), \alpha, \beta)} y_{v, g, d} &\ge 1 && \forall v\in T^* \\ 
        \label{eqn:CLPDW_paths_02}
        \sum_{v\in T^*} \sum_{g \in \C(v, \alpha, \beta)} \sum_{d\in \EX_{\ge i+1}(L_i(g\cap S_i), \alpha, \beta)} g(e) \cdot y_{v, g, d} & \le b(e) && \forall e \in E_i,\\
        \label{eqn:CLPDW_paths_03}
        \sum_{v\in T^*} \sum_{g \in \C(v, \alpha, \beta)} \sum_{d\in \EX_{\ge i+1}(L_i(g\cap S_i), \alpha, \beta)}d(e) \cdot y_{v,g,d} & \le b(e) && \forall e \in E_{\ge i+1},\\
        \label{eqn:CLPDW_paths_04}
        y_{v, g, d} &\ge 0 && \forall v\in T^*, g\in \C(v, \alpha, \beta), \\
        & && d\in \EX_{\ge i+1}(L_i(g\cap S_i), \alpha, \beta) \notag 
\end{align}
\end{LPBox}

It is now sufficient to either find a solution of $\DW$ or
a hyperplane that separates $b$ from all $b'$ for
which $\DW$ is feasible, i.e., from $\BLP_{\ge i}(T^*, \alpha, \beta)$.
Since $\DW$ has only polynomially many
constraints, it is useful to consider the dual, which has polynomial dimension.
\begin{LPBox}[Dual of $\DW_{\ge i}(T^*, \alpha, \beta)$]
\begin{align}
    \min \sum_{e\in E_{\ge i}} b(e) \mu_e &- \sum_{v\in T^*} \pi_v \\
  \sum_{e\in E_i} g(e) \mu_e + \sum_{e\in E_{\ge i+1}} d(e) \mu_e &\ge \pi_v && \forall v\in T^*, g\in \C(v, \alpha, \beta), \label{eq:separation} \\
  & && d\in \EX_{\ge i+1}(L_i(g\cap S_i), \alpha, \beta) \notag \\
  \pi_v &\ge 0 && \forall v\in T^* \\
  \mu_e &\ge 0 && \forall e\in E_{\ge i}
\end{align}
\end{LPBox}

Our plan is to apply the Ellipsoid method to the dual using an approximate
separation problem, which we introduce below. 
\paragraph*{Separation problem.}
Given dual variables $\pi_v, \mu_e$
find $u^*\in T^*$, $g^*\in \C(u^*, \alpha, \beta)$,
and $d^*\in \BLP_{\ge i+1}(L_i(g\cap S_i), \alpha, \beta)$ such that 
\begin{equation*}
    \sum_{e\in E_i} g^*(e) \cdot \mu_e + \sum_{e\in E_{\ge i+1}} d^*(e) \cdot \mu_e < \pi_{u^*}
\end{equation*}
or determine that $\pi_v, \mu_e$
satisfy all constraints \eqref{eq:separation} in the dual of $\DW_{\ge i+1}(T^*, 1, 1)$.

\begin{lemma}\label{lem:sepDW}
    Given an algorithm for the separation problem and some $b\in [0, \beta]^{E_{\ge i}}$, we can find with polynomially many calls to the algorithm and polynomial running time overhead a solution to $\DW_{\ge i}(T^*, \alpha, \beta)$ or 
    find a feasible solution $\pi, \mu$ to the dual of $\DW_{\ge i}(T^*, 1, 1)$
    with negative objective (which proves
    that $\DW_{\ge i}(T^*, 1, 1)$ is infeasible).

    \begin{proof}
        We have $\C(v, 1, 1)\subseteq \C(v,\alpha,\beta)$, and we can swap the requirement of $d\in \EX_{\ge i+1}(L_i(g\cap S_i), \alpha, \beta)$ for $d^*\in \BLP_{\ge i+1} (L_i(g\cap S_i), \alpha, \beta)$, since these also form valid constraints for the dual. Hence, the feasible region $P$ of the dual of $\DW_{\ge i}(T^*, \alpha, \beta)$ is contained in the feasible region~$Q$ of the dual of~$\DW_{\ge i+1}(T^*, 1, 1)$. An algorithm for the separation problem can be seen as an \emph{approximate separation oracle} in the sense that, given dual variables as input, it either detects membership of the larger polyhedron $Q$ or outputs a hyperplane separating the input from the smaller polyhedron~$P$.
        
        Notice that the all-zero-vector is feasible for the dual. Furthermore, if $(\pi, \mu)$ is a feasible solution with objective $-1$, then $(c\cdot \pi, c\cdot \mu)$ is also feasible for the dual and has an objective value of $-c$ for any $c > 0$, i.e., the dual is unbounded. 
        It follows that $\CLP_{\ge i+1} (L_i(g\cap S_i), \alpha, \beta)$ is feasible if and only if its dual has no solution of value $-1$.
        Let 
        \[P'=\{(\pi, \mu): (\pi, \mu) \in P, \ \sum_{e\in E_{\ge i}} b(e) \mu_e - \sum_{v\in T^*} \pi_v = -1\},\] \[Q'=\{(\pi, \mu): (\pi, \mu)\in Q, \ \sum_{e\in E_{\ge i}} b(e) \mu_e - \sum_{v\in T^*} \pi_v = -1\}.\] 
        We simulate the Ellipsoid method on $P'$ with the given approximate separation oracle (over $P'$ and $Q'$) in order to find a feasible point in $Q'$, if there exists one.
        Recall that in each iteration the Ellipsoid method presents
        a point $(\pi, \mu)$ and it expects us to either determine that $x\in P'$ or a hyperplane separating $x$ from $P'$. In each iteration, we apply the approximate separation oracle which either determines $(\pi, \mu) \in Q'$ or finds a hyperplane separating $(\pi, \mu)$ from $P'$.
        In the latter case, we continue the Ellipsoid method adhering to its requirements. In the former case we can stop the algorithm, since we found the required solution. Note that we terminate earlier than one would normally when running Ellipsoid on $P'$,
        since we cannot actually decide whether $(\pi, \mu)\in P'$ or not.

        \begin{comment}
        By applying the updating step of the original ellipsoid method, we maintain an ellipsoid that is known to contain $P'$, just as the original one does. Notice that this is not necessarily the case for $Q'$, parts of it could be cut off by the separating hyperplanes at some point. 
        After a polynomial number of iterations (i.e. polynomially many calls to the algorithm for the original separation problem), 
        \noteSM{don't we need a lower bound on the volume of P' here?}
        the adapted ellipsoid method either finds a point $(\pi, \mu) \in Q'$ or stops because we attain infeasibility of the separating hyperplanes for $P'$.
        \noteSM{Stopping argument? Again via volume argument?}
        \noteSM{TBC}           
        \end{comment}
        
        As mentioned above, the existence of $(\pi, \mu) \in Q'$ implies infeasibility of $\DW_{\ge i}(T^*, 1, 1)$. If such a point can not be found, it is enough to consider the primal variables that are encountered as dual constraints when solving the dual. Since we only had polynomially many calls to the approximate separation oracle, one can solve the primal restricted to these variables.
    \end{proof}
\end{lemma}

\begin{comment}
    \paragraph{Duals and approximate separation.}
    A typical approach to solve an implicitly given linear program with an exponential number of variables and polynomial number of constraints is to run the Ellipsoid method
   on the dual by repeatedly solve the separation problem. 
   It is enough to consider the primal variables that are encountered
   as dual constraints when solving the dual. Restricting to these variables one can solve the primal then. 
   In our case, we will only be able to solve the separation problem approximately. The above approach, however, can still be useful.
   This as well as the properties we need from the multilinear extension are summarized in the lemma below.
\end{comment}
    
Assuming we can solve the separation problem efficiently, this implies \Cref{thm:linprog}. Indeed, if the algorithm from \Cref{lem:sepDW} is not successful, then we find a feasible solution $\pi$, $\mu_e$ for the dual of $\DW_{\ge i}(T^*, 1, 1)$ such that $\sum_{e\in E_{\ge i}} b(e)\mu_e - \sum_{v\in T^*} \pi_v < 0$. Notice that for any $b'\in \BLP_{\ge i}(T^*, 1, 1)$ the dual has no negative solutions, thus $\sum_{e\in E_{\ge i}} b'(e)\mu_e - \sum_{v\in T^*} \pi_v \ge 0$. Hence, this provides us with a hyperplane that separates $b$ from $\BLP_{\ge i}(T^*, 1, 1)$ as required in \Cref{thm:linprog}.

\subsection{Separation via multilinear extension}
Our goal is now to devise an algorithm that solves
the separation problem stated above.
We will formulate a continuous relaxation of the separation problem,
solve and round it, but in order to do so we need to first introduce
the concept of \emph{multilinear extension}.
Let $f: \{0,1\}^n \rightarrow \R_{\ge 0}$ be a submodular function.
We want to extend $f: \{0,1\}^n \rightarrow \R$ to the domain $[0,1]^{n}$, since we will be working with continuous relaxations.
There are several natural extensions, one
of which is known as the \emph{multilinear extension}.
The multilinear extension $F: [0,1]^n \rightarrow \R$ is defined by
\begin{equation*}
    F(x) = \sum_{S \subseteq [n]} f(S) \prod_{i \in S} x_i \prod_{j \notin S} (1-x_j)  \ .
\end{equation*}
It is equivalent to set $F(x)=\mathbb{E}[f(X)]$ where $X$ is a random set with elements appearing independently with probabilities $x_i$. 
Notice that the definition of $F$ involves summing over all subsets of $[n]$, but a good approximation can be computed by sampling from the said distribution.
Although $F$ is not convex, there are strong results for
approximately maximizing it over polytopes, most notable
the continuous Greedy algorithm~\cite{vondrak2008optimal}.
We will need a non-standard variant that
is summarized in the lemma below.
\begin{lemma}\label{lem:cgreedy}
    Let $P \subseteq Q\subseteq [0, 1]^n$ be downward-closed polytopes\footnote{A polyope is downward-closed if for any $x$ in the polytope and any $y$ with $0 \le y_i \le x_i$ for each component $i$, $y$ is also in the polytope.}
    and suppose that we can approximately optimize over $Q$ in the following way: given some $c\in \mathbb R^n$ we can find some $x\in Q$
    such that $c^T x \ge c^T y$ for all $y\in P$ or determine that $P$ is empty. Let $F$ be the multilinear extension of a monotone submodular function $f$ and let $x^* = \mathrm{argmax}_{x\in P}F(x)$. Assume furthermore that $f(\{i\}) \le F(x^*)$ for any element $i\in [n]$. Then, there is a polynomial time Las Vegas algorithm that finds some $x\in Q$ such that
    $F(x) \ge (1 - 1/e - \epsilon) F(x^*)$.
\end{lemma}
The proof of the lemma is almost identical to the classical analysis
of the continuous Greedy algorithm~\cite{vondrak2008optimal}, see Appendix~\ref{app:continuous_greedy}.

We will now reformulate the separation problem as maximizing a submodular
function over a polytope.
Assume that $\pi_v, \mu_e$ do not satisfy all constraints \eqref{eq:separation} for the dual of
$\DW_{\ge i+1}(T^*, 1, 1)$, i.e., there exists some $u^*\in T^*$, $g^* \in \C(u^*, 1, 1)$, and $d^*\in \BLP_{\ge i+1}(L_i(g\cap S_i),1,1)$ with
\begin{equation*}
    \sum_{e\in E_i} g^*(e) \mu_e + \sum_{e\in E_{\ge i+1}} d^*(e) \mu_e < \pi_{u^*} \ .
\end{equation*}
We can assume that we know $u^*$ through guessing. Our goal is to approximately
find $g^*$ and $d^*$.
Using the definition of a configuration, the values $g^*$ and $d^*$ are feasible for the following system and have a value of at least $1$.
\begin{align}
    \max \ & f_{u^*}(g \cap \delta(u^*)) \\
    \sum_{e\in E_{i}} g(e) \mu_e + \sum_{e\in E_{\ge i+1}} d(e) \mu_e &\le \pi_{u^*} \label{eq:separation-constraint-1} \\
    g(\delta^-(v)) &= g(\delta^+(v)) &&\forall v\in V_i\setminus (S_i \cup T_i) \\
    d &\in \BLP_{\ge i+1}(L_i(g \cap S_i), 1, 1) \label{eq:separation-rec-1} \\
    g(e) &\in \{0, 1\} &&\forall e\in E_i
\end{align}
We need to find a solution of value at least $1/\alpha$ and we can make use of
congestion $\beta$.
For technical reasons we will assume without loss of generality that $f_{u^*}$
is upper bounded by $1$, which is possible by replacing $f_{u^*}$ with
$f_{u^*}'(S) = \min\{1, f_{u^*}(S)\}$, since this preserves monotonicity, submodularity and the existence of a solution of value $1$.
We further assume that each individual element has a small value, more precisely,
$f_{u^*}(\{e\}) \le 1/20$ for each $e\in \delta(u^*)$. The careful reader
will have noticed that the reduction in \Cref{sec:Reduction} anyway guarantees
this property. However, it is also easy to establish this assumption without
adding more technical restrictions to the definition of the augmentation problem\footnote{
If $g^*(e) = 1$ for some $e\in\delta(u^*)$ with $f_{u^*}(\{e\}) > 1/40$
we guess $e$ as well as the source $s\in S_i$, from which the flow on $e$ comes (in an arbitrary path decomposition of $g^*$. We compute a shortest path from $s$ to $e$
with weights $\mu_e$ and we compute some $d\in B_{\ge i+1}(L_i(\{s\}), \alpha, \beta)$ minimizing $\sum_{e\in E_{\ge i+1}} d(e) \mu_e$ using Ellipsoid method.
This yields a solution of the separation problem for $\alpha \ge 1/40$.
In the other case we can remove all $e\in \delta(u^*)$ with $f_{u^*}(\{e\}) > 1/40$, since they are not used.
}.

By \Cref{lem:separate} we can replace Constraint~\eqref{eq:separation-rec-1}
with the following equivalent constraints and additional variables.
The main idea is that we introduce a binary variable $y_s\in \{0, 1\}$ for each $s\in S_i$, which describes whether $s\in g\cap S_i$.
\begin{align*}
    d_s &\in y_s \cdot \BLP_{\ge i+1}(L_i(\{s\}), 1, 1) &&\forall s\in S_i \\
    y_s &= g(\delta^+(s)) &&\forall s\in S_i \\
    y_s &\in \{0, 1\} &&\forall s\in S_i \\
    d(e) &= \sum_{s\in S_i} d_s(e) &&\forall e\in E_{\ge i+1} \\
    d(e) &\in [0, 1] &&\forall e\in E_{\ge i+1} \\
    d_s(e) &\in [0, 1] &&\forall s\in S_i, e\in E_{\ge i+1}
\end{align*}
In order to find an approximate solution
we consider the continuous relaxation using the multi-linear extension as described below.
\begin{LPBox}[Relaxed separation problem $\SEP(u^*, \alpha, \beta)$ with multilinear extension]
\begin{align}
    \label{eq:sep-obj}
    \max \ & F_{u^*}(\ (g(e))_{e\in \delta(u^*)} \ ) \\
    \sum_{e\in E_i} g(e) \mu_e + \sum_{e\in E_{\ge i+1}} \sum_{s\in S_i} d_s(e) \mu_e &\le \pi_{u^*} \label{eq:separation-constraint} \\
    g(\delta^-(v)) &= g(\delta^+(v)) &&\forall v\in V_i\setminus (S_i \cup T_i) \\
    \label{eq:separation-rec} d_s &\in y_s \cdot \BLP_{\ge i+1}(L_i(\{s\}), \alpha, \beta) &&\forall s\in S_i \\
    y_s &= g(\delta^+(s)) &&\forall s\in S_i \\
    \sum_{s\in S_i} d_s(e) &\le 1 &&\forall e\in E_{\ge i+1} \\
    g(e) &\in [0, 1] &&\forall e\in E_i \\
    y_s &\in [0, 1] &&\forall s\in S_i \\
    d_s(e) &\in [0, 1] &&\forall s\in S_i, e\in E_{\ge i+1}
\end{align}
\end{LPBox}
Here $F_{u^*}$ is the multilinear extension of $f_{u^*}$. 
If $F_{u^*}$ is replaced by a linear objective, then we can
solve $\SEP$ approximately in the following sense.
\begin{lemma}
    Consider the system $\LINSEP(u^*, \alpha, \beta, c)$, where
    \eqref{eq:sep-obj} in $\SEP(u^*, \alpha, \beta)$ is replaced by a linear objective $\max \sum_{e\in\delta(u^*)} c_e g(e)$ for $c\in \mathbb R^{\delta(u^*)}$.

    Let $\alpha, \beta\in\mathbb N$.
    We can find a solution for $\LINSEP(u^*, \alpha, \beta, c)$ with value
    at least the optimum of $\LINSEP(u^*, 1, 1, c)$ using polynomial running time overhead and a polynomial number of queries that for a given $s\in S_i$ and $d_s\in [0, \beta]^{E_{\ge i+1}}$ determine either $d_s\in \BLP_{\ge i+1}(L_i(\{s\}), \alpha, \beta)$ or return a hyperplane separating $d_s$ from
    $\BLP_{\ge i+1}(L_i(\{s\}), 1, 1)$.
\end{lemma}
\begin{proof}
    We will run the Ellipsoid method to optimize over $\LINSEP$:
    similar to the proof of \Cref{lem:sepDW} we simulate Ellipsoid
    on $\LINSEP(u^*, 1, 1, c)$ and terminate once the given solution is feasible for
    $\LINSEP(u^*, \alpha, \beta, c)$.
    Optimization is reduced to feasibility testing by performing a binary
    search over the optimum.
    
    It remains to separate (approximately) over
    Constraint~\eqref{eq:separation-rec}.
    The query access from the premise can be lifted to this task:
    Fix some $s\in S_i$ and $d_s, y_s$. 
    If $d_s, y_s$ are all zero, then $d_s\in y_s \cdot \BLP_{\ge i+1}(L_i(\{s\}, \alpha, \beta)$.
    If $y_s = 0$ and $d_s(e) > 0$
    for some $e$, we have that
    \begin{equation*}
        d_s(e) > y_s \text{ and } d'_s(e) \le y'_s
    \text{ for all } d'_s \in y'_s \cdot \BLP_{\ge i+1}(L_i(\{s\}, 1, 1) \ ,
    \end{equation*}
    which serves as a separating hyperplane.
    Assume now that $y_s > 0$.
%    Let $d_s=\sum_{r\in \BLP_{\ge i+1}(L_i(\{s\}), 1, 1)}\lambda_i r_i$ for some $\lambda_i \geq 0$ with $\sum_i \lambda_i = y_s$, where $r_i$ are the vertices of~$P$. Then $\alpha \cdot (x,y)=(\alpha x, \alpha y) \in P'$ for any $\alpha >0$. Indeed, $\alpha x =\sum_i\alpha \lambda_i r_i$ with $\sum_i \alpha \lambda_i =\alpha y$. So both polyhedra are cones.
    Apply the query from the premise on $d_s / y_s$. Either $d_s / y_s \in \BLP_{\ge i+1}(L_i(\{s\}), \alpha, \beta)$, which implies that $d_s \in y_s \cdot \BLP_{\ge i+1}(L_i(\{s\}), \alpha, \beta)$ or we find a separating hyperplane $w \in \mathbb R^{E_{\ge i+1}}$, $W\in \mathbb R$ with $w^T (d_s/y_s) < W$ and $w^T d' \ge W$ for all $d'\in \BLP_{\ge i+1}(L_i(\{s\}), 1, 1)$. 
    It follows that
    \begin{equation*}
        w^T d_s - W y_s < 0 \text{ and } w^T d'_s - W y'_s \ge 0
        \text{ for all } d'_s \in y'_s \cdot \BLP_{\ge i+1}(L_i(\{s\}), \alpha, \beta) \ ,
    \end{equation*}
    which again serves as a separating hyperplane.  
\end{proof}
In order to employ the continuous Greedy algorithm (\Cref{lem:cgreedy}),
the feasible region needs to be downward closed,
which is not necessarily true here.
If, however, we project to the variables $g(e)$, $e\in\delta(u^*)$,
then the feasible region becomes downward closed.
This is easy to see by considering the path decomposition of the flow $g$, and the fact that the $\mu_e$ variables are non-negative.
Maximizing $F_{u^*}$ over the projection is equivalent to maximizing
it over the original feasible set. Hence, we can use 
the lemma to find such a solution with
value at least $1 - 1/e - \epsilon$ (assuming that the optimum is $1$).

As a final step, we need to round this continuous solution to an integral
one.
\begin{lemma}
    Let $\alpha \ge 40$ and $\beta \ge 10 \log(n)$.
    Given a solution of value $1 - 1/e - \epsilon$
    for $\SEP(u^*, \alpha, \beta)$ with $\epsilon > 0$ being a sufficiently small constant, there is a Las Vegas algorithm to 
    find a solution for the separation problem in polynomial time.
\end{lemma}
\begin{proof}
Let $g, y_s, d_s$ be the continuous solution of $\SEP(u^*, \alpha, \beta)$.
We compute a path decomposition of $g$ into
paths $\mathcal P$ from $S_i$ to $u^*$ and weights $\lambda_P, P\in\mathcal P$, that satisfy
$g = \sum_{P\in \mathcal P} \lambda_P \cdot \chi(P)$, where
$\chi(P)\in \{0,1\}^{E_i}$ is the characteristic
vector of path $P$. 
We partition $\mathcal P$ into sets $\mathcal P(e), e\in \delta(u^*)$, which are the paths with last edge being $e$.

Independently for each $e\in \delta(u^*)$ we
sample at most one path from $\mathcal P(e)$ such
that the probability of sampling $P$ is exactly $\lambda_P$.
Let $\mathcal P'$ be the resulting paths.
For simplicity of notation, we write
\begin{equation*}
    f_{u^*}(\mathcal Q) = f_{u^*}(\bigcup_{P\in \mathcal Q} P\cap \delta(u^*)) \ ,
\end{equation*}
for the submodular function value of some set of paths $\mathcal Q$.
Furthermore, for each $P\in \mathcal P'$ define
\begin{equation*}
  \Phi(P) = \sum_{e\in P} \mu_e + \sum_{e\in E_{\ge i + 1}} d_s(e)/y_s \cdot \mu_e \ ,
\end{equation*}
where $s$ is the source that $P$ originates from.
Then we have that
\begin{equation*}
    \mathbb E[\sum_{P\in\mathcal P'} \Phi(P)] = \sum_{e\in E_i} g(e) \mu_e + \sum_{e\in E_{\ge i+1}} \sum_{s\in S_i} y_s \cdot d_s(e) / y_s \cdot \mu_e \le \pi_{u^*} \ .
\end{equation*}
Thus, by Markov's inequality we get
\begin{equation*}
    \mathbb P[\sum_{P\in\mathcal P'} \Phi(P) > 10 \pi_{u^*}] \le 1/10 \ .
\end{equation*}
Further, let $p = \mathbb P[f_{u^*}(\mathcal P') < 1/2]$.
Since $f$ is bounded by~$1$ due to an earlier assumption, we have 
\begin{equation*}
    1 - 1/e - \epsilon \le F_{u^*}((g(e))_{e\in\delta(u^*)}) = \mathbb E[f_{u^*}(\mathcal P')] \le p/2 + (1 - p) \ .
\end{equation*}
This implies
\begin{equation*}
    p \le 2/e + 2\epsilon \le 0.75 \ ,
\end{equation*}
assuming $\epsilon$ (from the continuous Greedy algorithm) is choosen sufficiently small.
Define $k = \lceil \sum_{P\in\mathcal P'} \Phi(P) / \pi_{u^*} \rceil$.
With probability at least $0.15$ we have that 
\begin{equation}
    k \le 10 \quad \text{ and } \quad f_{u^*}(\mathcal P') \ge 1/2 \ .\label{eq:good-sample}
\end{equation}
We will show that if \eqref{eq:good-sample} holds, then we can recover a set of paths
$\mathcal P'' \subseteq \mathcal P'$ with $\sum_{P\in\mathcal P''} \Phi(P) < \pi_{u^*}$ and 
$f_{u^*}(\mathcal P'') \ge 1/\alpha$.

Partition $\mathcal P'$ into $k+1$ many sets $\mathcal P'_0, \dotsc, \mathcal P'_k$
where $|\mathcal P'_0| \le k$ and $\Phi(P'_i) < 1$ for all $i=1,2,\dotsc,k$:
for this, we greedily add paths to $\mathcal P'_1$ until $\Phi(P'_1) < 1$
would be violated with the next path. This path is added to $P'_0$. Then we repeat
the same with $\mathcal P'_2, \mathcal P'_3, \dotsc$ until all paths are placed
in one set. Since each iteration packs $\Phi(P)$ values of sum at least~$1$, the process must terminate after at most $k$ iterations.

Since $f$ is monotone submodular and in particular subadditive and since $f_{u^*}(\{e\}) \le 1/40$ by an earlier assumption, we have that
\begin{equation*}
    f_{u^*}(\mathcal P'_1) + \cdots + f_{u^*}(\mathcal P'_k) \ge f_{u^*}(\mathcal P') - f_{u^*}(\mathcal P'_0) \ge 1/2 - k \cdot 1/40 \ge 1/4 \ .
\end{equation*}
Thus, $f_{u^*}(\mathcal P'_i) \ge 1/(4k) \ge 1/40 = 1/\alpha$ for some $i\in \{1,2,\dotsc,k\}$.
Let $i$ be the index above and define $g' = \sum_{P\in \mathcal P'_i} \chi(P)$, where $\chi(P)$
is the characteristic vector of path $P$. In other words, $g'$ is the flow corresponding
to $\mathcal P'_i$.
We further define $d' = \sum_{s\in L_i(g'\cap S_i)} d_s/y_s$.

By the previous arguments we have that with probability at least $0.15$,
the function value $f_{u^*}(g'\cap S_i) \ge 1/40$.
It remains to show that with a high probability $g'$ has congestion at most $\beta$
and that $d'\in \BLP_{\ge i+1}(L_i(g'\cap S_i), \alpha, \beta)$.
For the former, we analyze the flow value on each edge $e\in E_i$ separately.
Note that
\begin{equation*}
    \mathbb E[g'(e)] \le \mathbb E[\ |\{ P\in \mathcal P' : e\in P\}| \ ]
    = g(e) \le 1 \ .
\end{equation*}
Further, $X = |\{ P\in \mathcal P' : e\in P\}|$ can be seen as a sum of independent $0/1$
random variables, one for each set $\mathcal P(e')$, $e'\in\delta(u^*)$.
Thus, we can apply a Chernoff bound on $X$, which implies that
\begin{equation*}
    \mathbb P[g'(e) > 10\log(n)] \le \mathbb P[X > 10\log(n)] \le 1/n^3 \ .
\end{equation*}
Now consider $d'$.
Since $d_s / y_s \in \CLP_{\ge i+1}(L_i(\{s\}),\alpha, \beta)$ for
all $s\in g'\cap S_i$ and $d'$ is their sum, due to \Cref{lem:separate}
it suffices to show that $d' \in [0, \beta]^{E_{\ge i+1}}$.
This we can argue in a similar way with a Chernoff bound:
let $e\in E_{\ge i+1}$. Then
\begin{equation*}
    \mathbb E[d'(e)] = \sum_{s\in S_i} \mathbb P[s \in g' \cap S_i] \cdot d_s(e) / y_s
    \le \sum_{s\in S_i} d_s(e) \le 1 \ .
\end{equation*}
Each term $d_s(e) / y_s$ can be seen as an independent random variable, which is bounded
by $1$ since $d_s\in y_s \cdot \CLP_{\ge i+1}(\{s\},1,1)$. Thus
\begin{equation*}
    \mathbb P[d'(e) > 10\log(n)] \le 1/n^3 \ .
\end{equation*}
To summarize, we have with probability $0.15 - |E_{\ge i}| /n^3$ that the solution we output
is a correct solution to the separation problem. This can
be boosted to high probability by repeating the random experiment.
\end{proof}

\section{Rounding the linear programming relaxation}
\label{sec:rounding}
In this section we will perform randomized rounding on a solution to 
the multi-level configuration LP, $\CLP_{\ge 1}(T^*, \alpha, \beta)$,
in order to arrive at
a solution for the augmentation problem. For convenience, we restate here the LP. We recall that $\BLP_{\ge i}(T^*, \alpha, \beta)$ is the set of feasible values $b\in [0,\beta]^{E_{\ge i}}$ for $\CLP_{\ge i}(T^*, \alpha, \beta)$.
\begin{LPBox}[Multi-level configuration LP, $\CLP_{\ge i}(T^*, \alpha, \beta)$]
%Let $b \in \CLP_{\ge i}(T^*, \alpha, \beta)$ if and only if
%there is a solution to
\begin{align}
        \label{eqn:CLP_paths_2_01}
        \sum_{g \in \C(v, \alpha, \beta)} x_{v, g} &\ge 1 && \forall v\in T^* \\ 
        \label{eqn:CLP_paths_2_03}
        \sum_{v\in T^*} \sum_{g \in \C(v, \alpha, \beta)} b_{v,g}(e) &\le b(e) && \forall e \in E_{\ge i}\\
        \label{eqn:CLP_paths_2_02}
        g(e) \cdot x_{v, g} &= b_{v,g}(e) && \forall v\in T^*, g\in \C(v, \alpha, \beta), \\
        \notag
        & && e \in E_i,\\
        \label{eqn:CLP_paths_2_04}
        (b_{v,g}(e))_{e\in E_{\ge i+1}} &\in x_{v, g} \cdot \BLP_{\ge i+1}(L_i(g\cap S_i), \alpha, \beta) && \forall v\in T^*, g\in \C(v, \alpha, \beta) \\
        \label{eqn:CLP_paths_2_05}
        x_{v, g} &\ge 0 && \forall v\in T^*, g\in \C(v, \alpha, \beta) \\ 
        b_{v, g}(e) &\ge 0 && \forall v\in T^*, g\in \C(v, \alpha, \beta), \\
        \notag
        & && e\in E_{\ge i}
\end{align}
\end{LPBox}
The rounding procedure will be defined recursively and its properties are summarized in the following lemma.
\begin{lemma}
\label{lem:rounding}
    Assume we are given a set of sinks $T^*\subseteq T_i$, and $b = (\gamma,\dotsc,\gamma)\in \BLP_{\ge i}(T^*, \alpha, \beta)$ with $\gamma \ge 6\beta \cdot \log^3 n$. Then, we can in polynomial time find an integral flow $\bar g$ in $G_i$ such that with high probability
    \begin{enumerate}
        \item flow $\bar g$ $\alpha$-covers every sink in $T^*$ with congestion $O(\gamma)$, and
        \item $(\bar\gamma,\bar\gamma,\dotsc,\bar\gamma) \in \BLP_{\ge i+1}(L_i(\bar g \cap S_i)), \alpha, \beta)$ for $\bar\gamma=\gamma\cdot \left(1+\frac{1}{\log n}\right)$.
    \end{enumerate}
\end{lemma}
\begin{proof}
    We proceed as follows. We let $x_{v,g}$ and $b_{v,g}(e)$ be the variables that attest $(\gamma,\gamma,\dotsc,\gamma)\in \BLP_{\ge i}(T^*, \alpha, \beta)$. We assume without loss of generality that \eqref{eqn:CLP_paths_2_01}
    holds with equality.
    Each sink $v\in T^*$ picks independently a flow (configuration) $g_v$ with probability $x_{v, g_v}$,
    which by the previous assumption is a valid probability distribution.  
    By Constraints~\eqref{eqn:CLP_paths_2_04}, we have that
    $b_{v,g_v}/ x_{v, g_v} \in \BLP_{\ge i+1}(L_i(g_v\cap S_i), \alpha, \beta)$ attested
    by some variables $x^{(v)}_{u,g_u}, b^{(v)}_{u,g_u}$ (corresponding
    to the conditions of $\CLP_{\ge i+1}$).
    We then define
    \begin{equation*}
      \bar g = \sum_{v\in T^*} g_v \ , \quad \bar x_{u, g_u} = \sum_{v\in T^*} x^{(v)}_{u, g_u} \ , \quad \text{ and }\quad \bar b_{u, g_u} = \sum_{v\in T^*} b^{(v)}_{u, g_u} \ .
    \end{equation*}
    We will show that with high probability $\bar g$ satisfies the properties of the lemma and
    $\bar x_{u, g_u},\bar b_{u,g_u}$ attest that $(\bar\gamma,\dotsc,\bar\gamma)\in \BLP_{\ge i+1}(\bar T^*, \alpha, \beta)$, where $\bar T^*=L_i(\bar g \cap S_i)$.
    
    For the first property, we need to analyze the congestion of $\bar g$. By Constraints~\eqref{eqn:CLP_paths_2_02} and~\eqref{eqn:CLP_paths_2_03} we obtain that the expected congestion on any edge $e$ is equal to 
    \begin{equation*}
    \mathbb E[\bar g(e)] = \mathbb E[\sum_{v\in T^*} g_v(e)]
        = \sum_{v\in T^*}\sum_{g\in \C(v, \alpha, \beta)} g(e) \cdot x_{v, g}\le 
        \sum_{v\in T^*}\sum_{g\in \C(v, \alpha, \beta)} b_{v, g} \le \gamma \ .
    \end{equation*}
    Further, the congestion is the sum of independent random variables, one for each $v\in T^*$, that
    are each bounded by $\beta$.
    Therefore, using a Chernoff bound, the probability that the congestion is more than $2\gamma$ is at most 
    \begin{equation*}
        \exp\left(-\frac{\gamma}{3\beta}\right)\le  \exp\left(-2\log^3 n)\right)\ .
    \end{equation*}
    
    Hence, with high probability, the congestion in $E_i$ at most $2\gamma$. For the second property,
    we verify Constraints \eqref{eqn:CLP_paths_2_01}, \eqref{eqn:CLP_paths_2_02}, \eqref{eqn:CLP_paths_2_03}, and  \eqref{eqn:CLP_paths_2_04} one by one, the last two being trivial to verify.
    
        \paragraph{Constraint~\eqref{eqn:CLP_paths_2_01}. } Let $u\in \bar T^*$ be one of our new sources. Let $g_w\in \C(w, \alpha, \beta)$ be a configuration that was selected for some $w\in T^*$ such that $u\in V(g_w)$ (it must exists by definition of $\bar T^*$). Then, by Constraint~\eqref{eqn:CLP_paths_2_04}, we have that that $b_{w, g_w}/x_{w, g_w}\in \BLP_{\ge i+1}(L_i(g_w\cap S_i), \alpha, \beta)$ attested by
        $x^{(w)}_{v, g}$, $b^{(w)}_{v, g}$. Due to Constraint~\eqref{eqn:CLP_paths_2_01} in
        $\CLP_{\ge i+1}$ we have
    \begin{equation*}
        \sum_{g\in \C(u, \alpha, \beta)} \bar x_{u, g} \ge \sum_{g\in \C(u, \alpha, \beta)} x^{(w)}_{u, g}\ge 1\ .
    \end{equation*}
    \paragraph{Constraints~\eqref{eqn:CLP_paths_2_02} and  \eqref{eqn:CLP_paths_2_04}.}
    Notice that $\bar x_{u, g}$, $\bar b_{u, g}$
    are the sum of variables that each satisfy~\eqref{eqn:CLP_paths_2_02} and \eqref{eqn:CLP_paths_2_04}.
    It is easy to see that these constraints remain
    satisfied under taking the sum of feasible solutions (since those constraints define polyhedral cones).
   
    \paragraph*{Constraint~\eqref{eqn:CLP_paths_2_03}.}
    Notice that
    \begin{align*}
        \mathbb E[\sum_{u\in \bar T^*}\sum_{g\in \C(u, \alpha, \beta)} \bar b_{u, g}(e) ]&\le \sum_{v\in T^*}\sum_{\substack{g\in \C(v, \alpha, \beta) \\ x_{v, g} > 0}} x_{v, g}\cdot \frac{b_{v, g}(e)}{x_{v, g}}\\
        &\le \sum_{v\in T^*}\sum_{g\in \C(v, \alpha, \beta)} b_{v, g}(e) \\
        &\le \gamma \ ,
    \end{align*}
    where the first inequality is obtained by definition of our sampling procedure and the last inequality by Constraint \eqref{eqn:CLP_paths_2_03} in $\CLP_{\ge i}$. Second, we notice that the random variable $\sum_{u\in \bar T^*}\sum_{g\in \C(u, \alpha, \beta)} \bar b_{u, g}(e)$ is a sum of independent random variables, one for each $u\in T^*$ that take a value $b_{u, g}(e)/x_{u, g}$ for some configuration $g$. By definition of $\BLP_{\ge i+1}$, we also have the constraint that $b_{u, g}(e)/x_{u, g}\le \beta$ for all $u$ and $g$. Hence the random variable $\sum_{u\in \bar T^*}\sum_{g\in \C(u, \alpha, \beta)} \bar b_{u, g}(e)$ is a sum of independent random variables, all bounded in absolute value by $\beta$, and of total expectation at most $\gamma$. 
   By a standard Chernoff bound, we have
    \begin{equation*}
       \mathbb P\left[ \sum_{u\in \bar T^*}\sum_{g\in \C(u, \alpha, \beta)} \bar b_{u, g}(e) \ge \gamma \cdot \left(1+1/\log n\right)\right]\le \exp\left(-\frac{\gamma}{2 \beta \log^2 n}\right) \le \exp(-3 \log n) \ .
    \end{equation*}
    Hence, with high probability, Constraint \eqref{eqn:CLP_paths_2_03} is satisfied as well.
\end{proof}

Now we can solve the augmentation problem by applying \Cref{lem:rounding} iteratively. If we have an LP solution with coverage $\alpha$ and congestion $\beta$ for an $h$-level instance, this yields an $(\alpha, O(\beta \log^3(n)\cdot (1+1/\log n)^h)$-approximate solution. For any $h = O(\log n)$, this is a $(\alpha,O(\beta \log^3 n))$-approximate solution.
Hence, \Cref{thm:linprog} and \Cref{lem:rounding} imply \Cref{thm:main-tech}.

\bibliographystyle{plain}
\bibliography{refs}

\begin{thebibliography}{10}

\bibitem{alon2016probabilistic}
Noga Alon and Joel~H Spencer.
\newblock {\em The probabilistic method}.
\newblock John Wiley \& Sons, 2016.

\bibitem{annamalai2017combinatorial}
Chidambaram Annamalai, Christos Kalaitzis, and Ola Svensson.
\newblock Combinatorial algorithm for restricted max-min fair allocation.
\newblock {\em ACM Transactions on Algorithms}, 13(3):1--28, 2017.

\bibitem{DBLP:journals/talg/AsadpourFS12}
Arash Asadpour, Uriel Feige, and Amin Saberi.
\newblock Santa claus meets hypergraph matchings.
\newblock {\em ACM Transactions on Algorithms}, 8(3):24:1--24:9, 2012.

\bibitem{bamas2021submodular}
{\'E}tienne Bamas, Paritosh Garg, and Lars Rohwedder.
\newblock The submodular santa claus problem in the restricted assignment case.
\newblock In {\em Proceedings of ICALP}, pages 22:1--22:18, 2021.

\bibitem{bamas2024santa}
{\'E}tienne Bamas, Alexander Lindermayr, Nicole Megow, Lars Rohwedder, and Jens
  Schl{\"o}ter.
\newblock Santa claus meets makespan and matroids: Algorithms and reductions.
\newblock In {\em Proceedings of SODA}, pages 2829--2860, 2024.

\bibitem{bamas2023better}
{\'E}tienne Bamas and Lars Rohwedder.
\newblock Better trees for santa claus.
\newblock In {\em Proceedings of the 55th Annual ACM Symposium on Theory of
  Computing}, pages 1862--1875, 2023.

\bibitem{bansal2006santa}
Nikhil Bansal and Maxim Sviridenko.
\newblock The santa claus problem.
\newblock In {\em Proceedings of {STOC}}, pages 31--40, 2006.

\bibitem{barman_et_al:LIPIcs.ESA.2020.11}
Siddharth Barman, Umang Bhaskar, Anand Krishna, and Ranjani~G. Sundaram.
\newblock {Tight Approximation Algorithms for p-Mean Welfare Under Subadditive
  Valuations}.
\newblock In {\em Proceedings of ESA}, pages 11:1--11:17, 2020.

\bibitem{bateni2009maxmin}
Mohammad~Hossein Bateni, Moses Charikar, and Venkatesan Guruswami.
\newblock Maxmin allocation via degree lower-bounded arborescences.
\newblock In {\em Proceedings of STOC}, pages 543--552, 2009.

\bibitem{calinescu2011maximizing}
Gruia Calinescu, Chandra Chekuri, Martin Pal, and Jan Vondr{\'a}k.
\newblock Maximizing a monotone submodular function subject to a matroid
  constraint.
\newblock {\em SIAM Journal on Computing}, 40(6):1740--1766, 2011.

\bibitem{DBLP:conf/focs/ChakrabartyCK09}
Deeparnab Chakrabarty, Julia Chuzhoy, and Sanjeev Khanna.
\newblock On allocating goods to maximize fairness.
\newblock In {\em Proceedings of {FOCS}}, pages 107--116, 2009.

\bibitem{dantzig1960decomposition}
George~B Dantzig and Philip Wolfe.
\newblock Decomposition principle for linear programs.
\newblock {\em Operations research}, 8(1):101--111, 1960.

\bibitem{davies2020tale}
Sami Davies, Thomas Rothvoss, and Yihao Zhang.
\newblock A tale of santa claus, hypergraphs and matroids.
\newblock In {\em Proceedings of {SODA}}, pages 2748--2757, 2020.

\bibitem{epstein1999approximation}
Leah Epstein and Ji{\v{r}}{\'\i} Sgall.
\newblock Approximation schemes for scheduling on uniformly related and
  identical parallel machines.
\newblock In {\em Proceedings of {ESA}}, pages 151--162, 1999.

\bibitem{feige2008allocations}
Uriel Feige.
\newblock On allocations that maximize fairness.
\newblock In {\em Proceedings of {SODA}}, pages 287--293, 2008.

\bibitem{feige2009maximizing}
Uriel Feige.
\newblock On maximizing welfare when utility functions are subadditive.
\newblock {\em SIAM Journal on Computing}, 39(1):122--142, 2009.

\bibitem{goemans2009approximating}
Michel~X Goemans, Nicholas~JA Harvey, Satoru Iwata, and Vahab Mirrokni.
\newblock Approximating submodular functions everywhere.
\newblock In {\em Proceedings of {SODA}}, pages 535--544, 2009.

\bibitem{DBLP:journals/corr/abs-2211-08381}
Sungjin Im, Benjamin Moseley, Hung~Q. Ngo, Kirk Pruhs, and Alireza Samadian.
\newblock Optimizing polymatroid functions.
\newblock {\em CoRR}, abs/2211.08381, 2022.

\bibitem{krause2009simultaneous}
Andreas Krause, Ram Rajagopal, Anupam Gupta, and Carlos Guestrin.
\newblock Simultaneous placement and scheduling of sensors.
\newblock In {\em Proceedings of IPSN}, pages 181--192, 2009.

\bibitem{DBLP:journals/mp/LenstraST90}
Jan~Karel Lenstra, David~B. Shmoys, and {\'{E}}va Tardos.
\newblock Approximation algorithms for scheduling unrelated parallel machines.
\newblock {\em Mathematical Programming}, 46:259--271, 1990.

\bibitem{li2022constant}
Wenzheng Li and Jan Vondr{\'a}k.
\newblock A constant-factor approximation algorithm for nash social welfare
  with submodular valuations.
\newblock In {\em Proceedings of FOCS}, pages 25--36, 2022.

\bibitem{PolacekS2015}
Luk{\'a}{\v{s}} Pol{\'a}{\v{c}}ek and Ola Svensson.
\newblock Quasi-polynomial local search for restricted max-min fair allocation.
\newblock {\em ACM Transactions on Algorithms}, 12(2):1--13, 2015.

\bibitem{svitkina2011submodular}
Zoya Svitkina and Lisa Fleischer.
\newblock Submodular approximation: Sampling-based algorithms and lower bounds.
\newblock {\em SIAM Journal on Computing}, 40(6):1715--1737, 2011.

\bibitem{vondrak2008optimal}
Jan Vondr{\'a}k.
\newblock Optimal approximation for the submodular welfare problem in the value
  oracle model.
\newblock In {\em Proceedings of {STOC}}, pages 67--74, 2008.

\bibitem{woeginger1997polynomial}
Gerhard~J Woeginger.
\newblock A polynomial-time approximation scheme for maximizing the minimum
  machine completion time.
\newblock {\em Operations Research Letters}, 20(4):149--154, 1997.

\end{thebibliography}

\appendix

\section{Continuous greedy with approximate separation}
\label{app:continuous_greedy}
In this section, we prove \Cref{lem:cgreedy}. Let $P\subseteq Q\subseteq \mathbb [0,1]^n$ be two polyhedra which are downward-closed.

Let $F$ be the multilinear relaxation of a monotone submodular function $f$. We also assume that for any element $i$ in the ground set, we have that $f(\{i\})\le F(x^*)$, where $x^*$ is the point in $P$ maximizing $F(x^*)$.

We show that we can obtain with high probability, in polynomial time for any fixed $\epsilon>0$, a point $y\in Q$ such that $F(y)\ge (1-1/e-\epsilon)\cdot F(x^*)$.

The proof is an easy modification of \cite{calinescu2011maximizing}, which we repeat here for completeness. The algorithm is as follows.

\begin{enumerate}
    \item Let $\delta=1/(10n^2)$, and let $t=0$, $y(0)=0$.
    \item Let $R(t)$ contain each $j\in [n]$ independently with probability $y_j(t)$. For all $j\in [n]$, we let $w_j(t)$ be an estimate of 
    \begin{equation*}
        \mathbb E[f(j\mid R(t))]
    \end{equation*} by taking the average over $\frac{10}{\delta^2}(1+\ln n)$ independent samples of $R(t)$. We denote by $\textbf{E}_t$ the vector whose $j$-th coordinate is equal to $\mathbb E[f(j\mid R(t))]$. We also aggregate the $w_j(t)$ into a single vector $w(t)$. 
    \item Find $y\in Q$ such that $w(t)^Ty\ge w(t)^Tx$ for all $x\in P$. We can find such a point by the assumption in \Cref{lem:cgreedy}. Set
    \begin{equation*}
        y(t+\delta)=y(t)+\delta \cdot y\ .
    \end{equation*}
    \item If $t<1$, return to step 2, otherwise output $y(1)$.  
\end{enumerate}

Note that the output $y(1)$ is a convex combination of points in $Q$ (recall that $y(0)=0\in Q$ since $Q$ is downward-closed). Hence, we have $y(1)\in Q$.

We use essentially the same arguments as in \cite{calinescu2011maximizing}. We start by the first key lemma.

\begin{lemma}
\label{lem:continuous_greedy}
    Let $y\in [0,1]^n$ and let $R$ be a random set containing each $j$ independently with probability $y_j$. Then
    \begin{equation*}
        F(x^*)\le F(y)+\max_{y'\in P} \sum_{j\in [n]}y'_j\cdot \mathbb E[f(j\mid R(t))]\ .
    \end{equation*}
\end{lemma}
\begin{proof}
    We can write by submodularity, for any set $R$, and any set $O$,
    \begin{equation*}
        f(O) \le f(R)+\sum_{j\in O}f(j\mid R)\ .
    \end{equation*}
    By taking the expectation over the set $R$ containing each $j$ independently with probability $y_j$ and over the set $O$ containing each element $j$ independently with probability $x_j^*$, we obtain 
    \begin{align*}
        F(x^*) \le F(y)+\sum_{j\in [n]}x_j^*\cdot \mathbb E[f(j\mid R)] \le F(y)+\max_{y'\in P} \sum_{j\in [n]}y'_j \cdot \mathbb E[f(j\mid R)]\ ,
    \end{align*}
    where the inequality holds since $x^*\in P$. This concludes the proof.
\end{proof}

The second lemma essentially states that estimating the expectations with sampling does not loose much. Before proving this result, we state here an inequality that will be useful in the proof. 
\begin{theorem}[Theorem A.1.16 in \cite{alon2016probabilistic}]
\label{thm:concentration}
Let $X_i$, $1\le i\le k$ be independent random variable with $\mathbb E[X_i]=0$ and $|X_i|\le 1$ for all $i$, then 
    \begin{equation*}
        \mathbb P[|\sum_{i=1}^{k} X_i|>a]\le 2\exp(-a^2/(2k))\ .
    \end{equation*}
\end{theorem}

\begin{lemma}
\label{lem:continuous_greedy_estimates}
    With probability at least $1-1/\mathrm{poly}(n)$, for every $t$ the algorithm finds some $y\in Q$ such that 
    \begin{equation*}
       (\textbf{E}_t)^T y \ge (1-2n\delta)\cdot F(x^*)-F(y(t))\ .
    \end{equation*}
\end{lemma}
\begin{proof}
    Recall that $y$ is selected to (approximately) maximize $w(t)^Ty$ among feasible points in $Q$, where $w_j(t)$ is our estimate of $ (\textbf{E}_t)_j=\mathbb E[f(j\mid  R(t))]$. We say that an estimate $w_j(t)$ is \textit{bad} if $|w_j(t)-\mathbb E[f(j\mid  R(t))]|>\delta \cdot F(x^*)$. As in \cite{calinescu2011maximizing}, one can argue that with high probability there is no bad estimate during the whole run of the algorithm.

    Let $R_1,R_2,\ldots ,R_k$ be the $k=\frac{10}{\delta^2}(1+\ln n)$ samples used for the estimates $w_j(t)$, and let us denote by $X_i$ the random variable $X_i=(f(j\mid  R_i)-\mathbb E[f(j\mid R(t))])/F(x^*)$. First, by submodularity and our assumption in the beginning of this section, we always have 
    \begin{align*}
        |X_i|\le \frac{\max(f(j\mid  R_i),\mathbb E[f(j\mid R(t))])}{F(x^*)}
        \le \frac{f(j)}{F(x^*)} \le 1\ .
    \end{align*}
    Next, note that the estimate is bad exactly if 
    \begin{equation*}
        \left|\sum_{i=1}^k X_i \right|>\frac{10}{\delta^2}(1+\ln n)\cdot \delta \ .
    \end{equation*}
    By applying \Cref{thm:concentration}, the probability of this happening is at most
    \begin{equation*}
        2\exp\left(-\frac{5}{\delta^2}(1+\ln n)\cdot \delta\right) = 2\exp\left(-5\ln (n)\right)
        \le n^{-4}\ .
    \end{equation*}
    By union bound over all $10n^2$ timesteps and all coordinates $j\in [n]$, with high probability all estimates are good.

    Now, let $y'\in P$ defined as
    \begin{equation*}
        y'=\text{argmax}_{y''\in P} (\textbf{E}_t)^Ty''\ ,
    \end{equation*}
    and let $M$ be this value. By \Cref{lem:continuous_greedy}, we have that 
    \begin{equation*}
        M\ge F(x^*)-F(y(t))\ .
    \end{equation*}
    Since all estimates are good, we also have that 
    \begin{equation*}
        \sum_{j\in [n]} y_j \cdot w_j(t)\ge \sum_{j\in [n]} y'_j \cdot w_j(t)\ge  M-\sum_{j\in [n]}y'_j |w_j(t)-\mathbb E[f(j\mid  R(t))]|\ge M-n\delta F(x^*)\ ,
    \end{equation*}
    where $y\in Q$ is the point chosen by the algorithm such that $(w(t))^Ty \ge (w(t))^Ty''$ for all $y''\in P$. Therefore, we obtain that 
    \begin{align*}
        \sum_{j\in [n]} y_j \cdot \mathbb E[f(j \mid R(t))] &\ge         \sum_{j\in [n]} y_j \cdot w_j(t) - \sum_{j\in n} y_j |w_j(t) - \mathbb E[f(j \mid R(t))] | \\
        &\ge \sum_{j\in [n]} y'_j \cdot w_j(t) - n\delta F(x^*) \\
        &\ge M- 2n\delta F(x^*)\\
        &\ge (1- 2n\delta)F(x^*)-F(y(t))\ ,
    \end{align*}
    as desired.
   % \noteLR{previous proof did not estimate the error of $y$, i.e., first line in the previous equation.}
\end{proof}

We can conclude with the main result we need.
\begin{lemma}
\label{lem:continuous_greedy_main}
With high probability, if $f$ is a monotone submodular function such that $F(x^*)\ge f(i)$ for all $i\in [n]$, the fractional solution $y$ found by the continuous greedy algorithm satisfies
\begin{equation*}
    F(y)\ge \left(1-1/e-1/(2n)\right)\cdot F(x^*)\ .
\end{equation*}
\end{lemma}
\begin{proof}
    Assume without loss of generality that $F(y(t)) \le F(x^*)$
    for all $t$, since otherwise the assertion follows immediately from
    the fact that $F(y) \ge F(y(t))$.
    The assumption is not trivial, since $x^*$ is in the set $P$,
    while $y(t)$ is allowed to be inside the bigger polytope $Q$.
    
    The algorithm starts with $F(y(0))=0$. We lower bound the increase in value at each step of the algorithm. The proof is the same as in \cite{calinescu2011maximizing}. Let $R(t)$ be the random set containing each element~$j$ independently with probability $y_j(t)$, and $D(t)$ the set containing each element independently with probability $\Delta_j(t)=y_j(t+\delta)-y_j(t)$. 
    We can easily see that 
    \begin{equation*}
        F(y(t+\delta))=\mathbb E[R(t+\delta)]\ge \mathbb E[f(R(t)\cup D(t))] \ .
    \end{equation*}
    This is because $R(t+\delta)$ contains $j$ with probability $y_j(t)+\Delta_j(t)$, while $R(t)\cup D(t)$ contains $j$ with smaller probability $1-(1-y_j(t))(1-\Delta_j(t))$. The two distributions can be coupled so that $R(t)\cup D(t)$ is a subset of $R(t+\delta)$, and we can conclude by the monotonicity of $f$.
    By denoting $y^*\in Q$ the direction in which we move $y(t)$, we can write
    \begin{align*}
        F(y(t+\delta))-F(y(t))&\ge \mathbb E[f(R(t)\cup D(t))-f(R(t))]\\
        &\ge \sum_{j\in [n]} \mathbb P[D(t)=\{j\}]\cdot \mathbb E[f(j\mid R(t))]\\
        &= \sum_{j\in [n]} (\delta y_j^*) \prod_{j'\neq j} (1-\delta y_{j'}^*) \cdot  \mathbb E[f(j\mid R(t))]\\
        &\ge  \sum_{j\in [n]} (\delta y_j^*) (1-n\delta) \cdot  \mathbb E[f(j\mid R(t))]\ .
    \end{align*}
    Using \Cref{lem:continuous_greedy_estimates} and $F(y(t))\le F(x^*)$, we obtain 
     \begin{align*}
        F(y(t+\delta))-F(y(t))&\ge \sum_{j\in [n]} (\delta y_j^*) (1-n\delta) \cdot  \mathbb E[f(j\mid R(t))]\\
        &\ge \delta (1-n\delta) ((1-2n\delta)\cdot F(x^*)-F(y(t)))\ge \delta ((1-3n\delta)\cdot F(x^*) -F(y(t))\ .
    \end{align*}
    Writing $\tilde F(x^*) =(1-3n\delta)\cdot F(x^*)$, we rearrange to get
    \begin{align*}
        \tilde F(x^*)-F(y(t+\delta))\le (1-\delta) (\tilde F(x^*)-F(y(t)))\ . 
    \end{align*}
    It follows now by induction that for any $k\ge 0$ (recall that $F(y(0))=0$), 
    \begin{equation*}
         \tilde F(x^*)-F(y(k\delta))\le (1-\delta)^k \tilde F(x^*)\ .
    \end{equation*}
    Hence, 
    \begin{align*}
        F(y(1))&\ge \tilde F(x^*) (1-(1-\delta)^{1/\delta})\\
        &\ge \tilde F(x^*) (1-1/e)\\
        &=(1-3n\delta)\cdot F(x^*)\cdot  (1-1/e)\\
        &\ge (1-1/e-1/(2n))\cdot F(x^*)\ ,
    \end{align*}
    which concludes the proof.
\end{proof}

\section{Proof of \Cref{lem:separate}}
\label{app:separate}
For convenience, we restate the assertion here:
   Let $T^*, T^{**}$ be disjoint sets of sinks
    and let $b\in [0, \beta]^{E_{\ge i}}$.
    Then
    $b \in \BLP_{\ge i}(T^* \cup T^{**}, \alpha, \beta)$
    if and only if there exist $b' + b'' = b$ with
    $b'\in \BLP_{\ge i}(T^*, \alpha, \beta)$ and $b''\in \BLP_{\ge i}(T^{**}, \alpha, \beta)$.

\begin{proof}
Let $(b, b_{v,g}, x_{b,g})\in\CLP_{\ge i}(T^* \cup T^{**}, \alpha, \beta)$
and assume without loss of generality that each constraint \eqref{eqn:CLP_paths_02} is tight.
Let
\begin{equation*}
    x'_{v, g} := \begin{cases}
        x_{v, g} &\text{ if } v\in T^* \\
        0 &\text{ otherwise. }
    \end{cases}
    \text{ and }
    b'_{v, g} := \begin{cases}
        b_{v, g} &\text{ if } v\in T^* \\
        0 &\text{ otherwise. }
    \end{cases}
\end{equation*}
Similarly, let
\begin{equation*}
    x''_{v, g} := \begin{cases}
        x_{v, g} &\text{ if } v\in T^{**} \\
        0 &\text{ otherwise. }
    \end{cases}
    \text{ and }
    b''_{v, g} := \begin{cases}
        b_{v, g} &\text{ if } v\in T^{**} \\
        0 &\text{ otherwise. }
    \end{cases}
\end{equation*}
Define $b'(e) = \sum_{v\in T^*}\sum_{g\in\C(v,\alpha,\beta)} b'_{v,g}(e)$
and $b''(e) = \sum_{v\in T^{**}}\sum_{g\in\C(v,\alpha,\beta)} b''_{v,g}(e)$.
Then $(b', b'_{v,g}, x'_{v,g})\in \CLP_{\ge i}(T^*,\alpha,\beta)$, $(b'', b''_{v,g}, x''_{v,g})\in \CLP_{\ge i}(T^{**},\alpha,\beta)$, and
$b = b' + b''$.

For the other direction, let $(b', b'_{v,g}, x'_{v,g})\in \CLP_{\ge i}(T^*,\alpha,\beta)$ and $(b'', b''_{v,g}, x''_{v,g})\in \CLP_{\ge i}(T^{**},\alpha,\beta)$. Then $(b' + b'', b'_{v,g} + b''_{v,g}, x'_{v,g} + x''_{v,g})\in \CLP_{\ge i}(T^* \cup T^{**}, \alpha, \beta)$.
\end{proof}

\end{document}